\documentclass[12pt]{article}
\usepackage{amssymb,amsmath,amsfonts,amsthm,lscape,xcolor,color,enumerate}
\usepackage[color,all]{xy}

\def\dd{{\rm d}}

\def\ham{{\cal {H}}_\Lambda}
\def\tX{ {X}_1}

\def\ff{g} 

\def\fff{g}

\def\mm{\mu}

\def\hh{\gamma} 

\theoremstyle{plain}
\newtheorem{theorem}{Theorem}
\newtheorem{corollary}[theorem]{Corollary}
\newtheorem{proposition}[theorem]{Proposition}
\newtheorem{lemma}[theorem]{Lemma}
\theoremstyle{definition}
\newtheorem{definition}[theorem]{Definition}

\newtheorem{example}[theorem]{Example}

\numberwithin{equation}{section}
\numberwithin{theorem}{section}

\input xy
\xyoption{all}
\xyoption{frame}

\begin{document}

\centerline{\Large {\bf    Lie--Hamilton systems on the plane:}}
\vskip 0.25cm
 \centerline{\Large {\bf properties, classification and  applications }}
\vskip 1cm
\centerline{ \sc{A. Ballesteros$^{{a}}$, A. Blasco$^{{a}}$, F.J. Herranz$^{{a}}$, J. de Lucas$^{{b}}$ and C. Sard\'on$^{c}$}}
\vskip 0.8cm
\centerline{$^{{a)}}$Department of Physics,  University of Burgos}
\smallskip
\centerline{09001, Burgos, Spain}

\medskip

\centerline{$^{{b)}}$Department of Mathematical Methods in Physics, University of Warsaw}
\smallskip
\centerline{ul. Ho\.za 74, 00-682, Warszawa, Poland}
\smallskip
\medskip

\centerline{$^{c)}$Fundamental Physics Department, University of Salamanca}
\smallskip
\centerline{ Plza. de la Merced s/n, 37008, Salamanca, Spain}
\smallskip
\medskip

\vskip 1cm

\begin{abstract}
We study Lie--Hamilton systems on the plane, i.e. systems of first-order differential equations describing the integral curves of a $t$-dependent vector field taking values in a finite-dimensional real Lie algebra of planar Hamiltonian vector fields with respect to a Poisson structure. We start with the 
  local classification of finite-dimensional real Lie algebras of vector fields on the plane obtained in [{A.~Gonz{\'a}lez-L{\'o}pez, N.~Kamran and P.J.~Olver}, {\em Proc. London Math. Soc.} {\bf 64}, 339 {(1992)}]
and we interpret their results as a  local classification of Lie systems.
 Moreover, by determining which of these real Lie algebras consist of Hamiltonian vector fields with respect to a Poisson structure, we provide the complete local classification of Lie--Hamilton systems on the plane. We present and study through our results new Lie--Hamilton systems of interest which are used to investigate relevant non-autonomous differential equations, e.g. we get explicit local diffeomorphisms between such systems. In particular, the  Milne--Pinney, second-order  Kummer--Schwarz, complex Riccati   and  Buchdahl equations as well as some  Lotka--Volterra and nonlinear biomathematical models are analysed from this Lie--Hamilton approach.
\end{abstract}

\bigskip


\medskip\noindent
 {{\bf MSC 2000:} 34A26 (Primary), 17B81, 34A34, 53Z05
(Secondary)}


\newpage


\section{Introduction}
 The relevance of time--dependent differential equations is undoubtable both from the mathematical viewpoint and also from their overwhelming applications. In this work we will get a deeper insight into a particular class of systems of differential equations, the so-called Lie systems,  which have drawn some attention during the past recent years since, for instance, the general solution for a Lie system can be obtained in terms of a superposition rule (see \cite{Dissertationes} and references therein). 

More explicitly, a {\it Lie system} is a system of first-order
differential equations describing the integral curves of
a $t$-dependent vector field taking values in a finite-dimensional real Lie algebra of vector fields, the {\it Vessiot--Guldberg Lie algebra} \cite{Dissertationes,LS}. This Lie algebra determines the main properties of Lie systems, e.g. Lie systems related to a solvable Vessiot--Guldberg Lie algebra of right-invariant vector fields on a Lie group are integrable \cite{CarRamGra}. Although Lie systems are a quite restricted class of differential equations \cite{Dissertationes,In72}, very recurrent systems appearing in the literature, e.g. most types of Riccati and Kummer--Schwarz equations, can {be studied through these systems }\cite{PW,CGL11}. In this paper, we aim to study {\it Lie--Hamilton  systems}~\cite{Dissertationes,CLS12Ham,BCHLS13Ham,ADR12}, which form a relevant  subclass of Lie systems. Our concern in them relies on their frequent appearance in classical mechanics and their special characteristics: integrability, symmetries and superposition rules \cite{BCHLS13Ham,ADR12,LS12,CGM00}.

A natural problem in the theory of Lie systems is the classification of Lie systems on a fixed manifold, which amounts to classifying finite-dimensional Lie algebras of vector fields on it.
Lie accomplished the local classification of finite-dimensional real Lie algebras of  vector fields on the real line. More precisely, he showed that each such a Lie algebra is locally diffeomorphic to a Lie subalgebra of $\langle \partial_x,x\partial_x,x^2\partial_x\rangle \simeq \mathfrak{sl}(2)$ on a neighborhood of each {\it generic point} $x_0$ of the Lie algebra~\cite{GKP92}. He also performed the local classification of finite-dimensional real Lie algebras of planar vector fields and started the study of the analogous problem on $\mathbb{C}^3$~\cite{1880,HA75}.

Lie's local classification on the plane presented some unclear points which were misunderstood by several authors during the following decades.
Later on, {A.~Gonz{\'a}lez-L{\'o}pez, N.~Kamran and P.J.~Olver} retook the problem and provided a clearer insight in \cite{GKP92}. Precisely, they proved that every non-zero Lie algebra of vector fields on the plane is locally diffeomorphic around each generic point to one of the finite-dimensional real Lie algebras given in Table \ref{table1} at the end of the work. For simplicity, we refer to this result as the {\it GKO classification}.

As every Vessiot--Guldberg Lie algebra on the plane is locally diffeomorphic to a Lie algebra of the GKO classification,  every Lie system on the plane is locally diffeomorphic to a Lie system taking values in a Vessiot--Guldberg Lie algebra within the GKO classification. So, the local properties of all Lie systems on the plane can be studied through the Lie systems related to the GKO classification. As a consequence, we say that the GKO classification gives the local classification of Lie systems on the plane.

The {\it minimal Lie algebra} of a Lie system is its smallest Vessiot--Guldberg Lie algebra \cite{Dissertationes}.
In this paper we analyse the general properties of minimal Lie algebras of Lie--Hamilton systems on the plane. We demonstrate that they are, around generic points, Lie algebras of Hamiltonian vector fields with respect to a symplectic structure. We also provide several results allowing us to determine their algebraic structure.

It is known that each Lie--Hamilton system on a manifold $N$ gives rise to a $t$-dependent Hamiltonian $h:(t,x)\in \mathbb{R}\times N\mapsto h_t(x)\in N$ whose functions $\{h_t\}_{t\in\mathbb{R}}$ generate a finite-dimensional Lie algebra of functions with a Lie bracket induced by certain Poisson structure: a {\it Lie--Hamilton algebra} \cite{CLS12Ham}. We obtain some findings concerning
the structure of the different Lie--Hamilton algebras of a Lie--Hamilton system. 

Based on the GKO classification and our previous achievements, we prove that a Lie algebra of Hamiltonian vector fields on the plane (with respect to a certain Poisson bivector) is locally diffeomorphic around a generic point to one of the twelve Lie algebras of Table \ref{table3}.
In this manner, we obtain the local classification of finite-dimensional Lie algebras of Hamiltonian vector fields on the plane. Subsequently, we provide the local classification of Lie--Hamilton algebras on the plane, namely we prove that the restriction of such a Lie algebra around a generic point (of the associated Lie algebra of Hamiltonian vector fields) is isomorphic to one of the Lie algebras indicated in Table \ref{table3}.  This is relevant to the theory of Lie--Hamilton systems because, for instance, the superposition rules and constants of motion of such systems can be obtained by applying the Poisson coalgebra approach  to such Lie algebras~\cite{CLS12Ham,BCHLS13Ham}.

Next, we detail some applications of our findings. By means of the GKO classification, we explain  that 
Milne--Pinney equations~\cite{PW, SIGMA} actually comprise {\em three} different systems (one of them  is  the harmonic oscillator with a $t$-dependent frequency). Likewise, we show   that second-order Kummer--Schwarz equations~\cite{CGL11}
 also cover {\em three} different systems and each of them is related to one of the Milne--Pinney equations  through a local diffeomorphism. 
Moreover, certain complex Riccati and Bernoulli equations \cite{Ca97,FMR10, Eg07}  are shown to be locally diffeomorphic   to only  one of the above three systems. This retrieves known results about second-order Kummer--Schwarz and Milne--Pinney equations and describes new relations between these systems and complex Riccati equations. Furthermore, we show how Buchdahl equations~\cite{Bu64, CSL05, CN10}, certain Lokta--Volterra systems~\cite{Tr96, JHL05, LlibreValls2} as well as some biological models  \cite{EK05} can be analysed through Lie--Hamilton systems. Indeed, we think that our techniques could be useful in different contexts. 

The structure of this paper goes as follows. Section 2 is devoted to introducing the fundamental definitions employed throughout the paper. In Section 3, we survey some basic facts about  the GKO classification. In Sections 4 and 5 we describe some new results on minimal Lie algebras and Lie--Hamilton algebras of functions on the plane. These two sections contain the necessary  theory to provide the local classification of Lie--Hamilton systems and their Lie--Hamilton algebras in Section 6.     Our main achievements are listed in Table \ref{table3}. To illustrate our results, we investigate in Section 7 some Lie--Hamilton $\mathfrak{sl}(2)$-systems on the plane, meanwhile applications to biological models are addressed in Section 8.
We conclude in Section 9 with a brief summary of the results here presented, together with some comments on possible future research work on the subject.


\section{Preliminaries}\label{BLHS}

Let us detail the notation and the most basic results to be used in the paper (see~\cite{CGL11,CLS12Ham, BCHLS13Ham,CGM00,CGM07,CGL10}
for details). We mostly assume mathematical objects to be smooth, real, and globally defined. This simplifies our presentation and is helpful in order to highlight its key points.

A Lie algebra is a  pair $(V,[\cdot,\cdot])$, where $V$ stands for a
real linear space equipped with a Lie
 bracket $[\cdot\,,\cdot]:V\times V\rightarrow V$. We define ${\rm
Lie}(\mathcal{B},V,[\cdot,\cdot])$ to be the smallest Lie subalgebra of $(V,[\cdot,\cdot])$ containing $\mathcal{B}$. When its meaning is clear, we write $V$ and  ${\rm Lie}(\mathcal{B})$ instead of $(V,[\cdot,\cdot])$ and ${\rm
Lie}(\mathcal{B},V,[\cdot,\cdot])$, respectively. Given two subsets $\mathcal{A},\mathcal{B}\subset V$, we write $[\mathcal{A},\mathcal{B}]$ for the linear space spanned by the Lie brackets between elements of $\mathcal{A}$ and $\mathcal{B}$. Given a Lie algebra $V$ of vector fields on a manifold $N$ and an open subset $U\subset N$, we define $V|_U$ to be the space of restrictions of the elements of $V$ to $U$. Note that $V|_U$ is still a Lie algebra of vector fields.

\begin{definition}
A {\it $t$-dependent vector field} on a manifold
$N$ is a mapping $X:\mathbb{R}\times N\rightarrow
TN$ such that $\tau\circ X=\pi$ for $\pi:(t,x)\in \mathbb{R}\times N\mapsto
x\in N$ and $\tau:
TN\rightarrow N$ being the tangent bundle projection related to $N$.
\end{definition}

Observe that every
$t$-dependent vector field $X$ gives rise to a family $\{X_t\}_{t\in\mathbb{R}}$ of
standard vector fields $X_t:x\in N\mapsto X(t,x)\in TN$ and vice
versa \cite{Dissertationes}.

\begin{definition}  The  {\it minimal Lie algebra} of a $t$-dependent vector field $X$ on $N$ is the smallest real Lie algebra, let us say
$V^X$, containing $\{X_t\}_{t\in\mathbb{R}}$, i.e. $V^X={\rm Lie}(\{X_t\}_{t\in\mathbb{R}},[\cdot,\cdot])$.
\end{definition}

\begin{definition}
An {\it integral curve} of a $t$-dependent vector field $X$ is an integral curve
$\gamma:\mathbb{R}\rightarrow\mathbb{R}\times N$ of its {\it suspension}, namely the vector field $\bar{ X}=\partial_t+X(t,x)$ on $\mathbb{R}\times N$.
\end{definition}

The integral curves of $X$ of the form $\gamma:t\in\mathbb{R}\rightarrow (t,x(t))\in\mathbb{R}\times N$ are such that $x(t)$ is a particular solution of the system of first-order differential equations in normal form
\begin{equation}\label{Sys}\nonumber
\frac{{\rm d} x}{{\rm d}  t}=X(t,x),
\end{equation}
the referred to as {\it associated system} of $X$. Conversely, given a system of first-order differential equations in normal form, we can define a $t$-dependent vector field $X$ whose integral curves of the form $t\mapsto (t,x(t))$ are such that $x(t)$ is a particular solution of such a system. This justifies to write $X$ for both a
$t$-dependent vector field and its associated system \cite{Dissertationes}.

\begin{definition}
A {\it Lie system} is a system $X$ whose $V^X$ is finite-dimensional.
\end{definition}

\begin{example} Consider the system of differential equations
\begin{equation}\label{Riccati2}
\frac{{\rm d} x}{{\rm d} t}=a_0(t)+a_1(t)x+a_2(t)(x^2-y^2),\qquad \frac{{\rm d} y}{{\rm d} t}=a_1(t)y+a_2(t)2xy,
\end{equation}
with $a_0(t),a_1(t),a_2(t)$ being arbitrary $t$-dependent real functions. This system is a particular type of planar Riccati equation briefly studied in \cite{Eg07}. By writing $z=x+{\rm i}y$, we find that (\ref{Riccati2}) is equivalent to
$$
\frac{{\rm d} z}{{\rm d} t}=a_0(t)+a_1(t)z+a_2(t)z^2,\qquad z\in\mathbb{C},
$$
which is a  particular type of complex Riccati equations, whose study has attracted some attention.  Particular solutions of periodic equations of this type have been investigated in \cite{Or12,Ca97} and other special types of complex Riccati equations appear in \cite{FMR10}.

Every particular solution $(x(t),y(t))$ of (\ref{Riccati2}) obeying that $y(t_0)=0$ at $t_0\in\mathbb{R}$ satisfies that $y(t)=0$ for every $t\in\mathbb{R}$. In such a case, $x(t)$ is a particular solution of a {\em real} Riccati equation \cite{PW}. This suggests us to restrict ourselves to studying (\ref{Riccati2}) on
$
\mathbb{R}^2_{y\neq 0}=\{(x,y)\,|\, y \neq 0\}\subset \mathbb{R}^2
$.

Let us show that (\ref{Riccati2}) on $\mathbb{R}^2_{y\neq 0}$ is a Lie system. This is related to the $t$-dependent vector field $X_t=a_0(t)X_1+a_1(t)X_2+a_2(t)X_3$, where
\begin{equation}
X_1= \frac{\partial}{\partial x},\qquad X_2= x\frac{\partial}{\partial x}+y\frac{\partial}{\partial y} ,\qquad X_3= (x^2-y^2)\frac{\partial}{\partial x}+2xy\frac{\partial}{\partial y}
\label{vectRiccati2}
\end{equation}
span a Vessiot--Guldberg real Lie algebra $V\simeq \mathfrak{sl}(2)$ (see ${\rm P}_2$ in Table \ref{table1}). Hence, $\{X_t\}_{t\in\mathbb{R}}\subset V^X\subset V$ and $V^X$ is finite-dimensional, which makes $X$ into a Lie system. It is worth noting that, to the best of our knowledge, this is the first time that it has been proved that complex Riccati equations with real coefficients and planar Riccati equations can be studied through Lie systems. Moreover, it can also  be demonstrated that complex Riccati equations with $t$-dependent complex coefficients can be investigated with a Lie system possessing a Vessiot--Guldberg Lie algebra isomorphic to P$_7\simeq \mathfrak{so}(3,1)$.
\end{example}

\begin{definition} A system $X$ is said to be a {\it Lie--Hamilton system} if
$V^X$ is a real finite-dimensional Lie algebra of Hamiltonian vector fields with respect to some Poisson bivector.
\end{definition}

\begin{definition} A {\it Lie--Hamiltonian structure} is a triple
$(N,\Lambda,h)$, where $\Lambda$ is a Poisson bivector and  $h:
(t,x)\in \mathbb{R}\times N\mapsto h_t(x)=h(t,x)\in  \mathbb{R}$ is such that
 $\ham\equiv{\rm Lie}(\{h_t\}_{t\in\mathbb{R}},\{\cdot,\cdot\}_\Lambda)$, with  $\{\cdot,\cdot\}_\Lambda$ being the Lie bracket induced by $\Lambda$ \cite{IV}, is
finite-dimensional.
\end{definition}

 \begin{theorem}\label{FunHam} A system $X$ on $N$ is a {\it Lie--Hamilton system} if and only
 if there exists a Lie--Hamiltonian structure $(N,\Lambda,h)$ such that each $X_t$, with $t\in\mathbb{R}$, is a Hamiltonian vector field for the function
$h_t$. In this case, we call $(\ham,\{\cdot,\cdot\}_\Lambda)$ a {\it Lie--Hamilton} algebra of $X$.
 \end{theorem}

Lie--Hamilton algebras play a relevant r\^ole in studying Lie--Hamilton systems, e.g. they are employed to calculate superposition rules and constants of motion for these systems more easily than by standard methods \cite{BCHLS13Ham}.

\begin{example} Let us show that planar Riccati equations (\ref{Riccati2}) with $V^X\simeq \mathfrak{sl}(2)$ are Lie--Hamilton systems and derive a Lie--Hamiltonian structure and its associated Lie--Hamilton algebra. We start by searching a symplectic form, let us say $\omega=f(x,y){\rm d}x\wedge \dd y$, turning $V^X$ into a Lie algebra of Hamiltonian vector fields with respect to it. To ensure that $X_1,X_2$ and $X_3$ given by (\ref{vectRiccati2}) are locally Hamiltonian vector fields with respect to $\omega$, we impose $\mathcal{L}_{X_i}\omega=0$ ($i=1,2,3$), where $\mathcal{L}_{X_i}\omega$ stands for the Lie derivative of $\omega$ relative to $X_i$. In coordinates, these conditions read
$$
\frac{\partial f}{\partial x}=0,\qquad x\frac{\partial f}{\partial x}+y\frac{\partial f}{\partial y}+2 f=0,\qquad
(x^2-y^2)\frac{\partial f}{\partial x}+2xy\frac{\partial f}{\partial y}+4xf=0.
$$
From the first equation $f=f(y)$. Using this in the second equation, we obtain a particular solution $f=y^{-2}$  (the third one is therefore automatically fulfilled), which leads to a closed and non-degenerate two-form on $\mathbb{R}^2_{y\neq 0}$, namely
\begin{equation}
\omega=\frac{{\rm d} x \wedge {\rm d} y}{y^2}  .
\label{aa}
\end{equation}
Using the relation $\iota_{X}\omega={\rm d}h$ among a Hamiltonian vector field $X$ and one of its corresponding Hamiltonian functions $h$, we observe that $X_1,X_2$ and $X_3$ are Hamiltonian vector fields with Hamiltonian functions
 \begin{equation}
h_1=-\frac 1y,\qquad h_2=-\frac xy,\qquad h_3=-\frac{x^2+y^2}{y},
\label{ab}
\end{equation}
 respectively. Since $X_1$, $X_2$ and $X_3$ are a basis for $V^X$, every element of $V^X$ is Hamiltonian with respect to $\omega$. If $\{\cdot,\cdot\}_\omega:C^\infty(\mathbb{R}^2_{y\neq 0})\times C^\infty(\mathbb{R}^2_{y\neq 0})\rightarrow C^\infty(\mathbb{R}^2_{
y\neq 0})$ stands for the Poisson bracket induced by $\omega$ (see \cite{IV}), then
\begin{equation}
\label{sl2Rh}
\{h_1,h_2\}_\omega=-h_1,\qquad \{h_1,h_3\}_\omega=-2h_2,\qquad \{h_2,h_3\}_\omega=-h_3.
\end{equation}
Hence, $\left(\mathbb{R}^2_{y\neq 0},\omega,h=a_0(t)h_1+a_1(t)h_2+a_2(t)h_3\right)$ is a Lie--Hamiltonian structure for $X$
and, as $V^X\simeq \mathfrak{sl}(2)$, then $(\mathcal{H}_\Lambda,\{\cdot,\cdot\}_\omega)\equiv (\langle h_1,h_2,h_3\rangle,\{\cdot,\cdot\}_\omega)$ is a Lie--Hamilton algebra for $X$ isomorphic to $\mathfrak{sl}(2)$.
\end{example}


\section{The  GKO  classification of real Lie algebras of  vector fields on the plane}

Let us summarise the main aspects and notions related to the GKO classification.

\begin{definition}\label{GenPon} Given a finite-dimensional Lie algebra $V$ of vector fields on a manifold $N$, we say that $\xi_0\in N$ is a {\it generic point} of $V$ when the rank of the generalised distribution
$$
\mathcal{D}^V_\xi=\{X(\xi)\mid X\in V\}\subset T_\xi N,\qquad \xi\in N,
$$ i.e. the function $r^V(\xi)=\dim \mathcal{D}^V_\xi$, is locally constant around $\xi_0$. We call {\it generic domain} or simply {\it domain} of $V$ the set of generic points of $V$.
\end{definition}

\begin{example} Consider the Lie algebra ${\rm I}_4=\langle X_1,X_2,X_3\rangle$ of vector fields on $\mathbb{R}^2$ detailed in Table \ref{table1}. By using the  expressions of $X_1$, $X_2$ and $X_3$ in coordinates, we see that $r^{{\rm I}_4}(x,y)$ equals the rank of the matrix
$$
\left(\begin{array}{ccc}
1&x&x^2\\
1&y&y^2\\
\end{array}
\right),
$$
which is two for every $(x,y)\in\mathbb{R}^2$ except for points with $y-x=0$, where the rank is one. So, the domain of I$_4$ is $
\mathbb{R}^2_{x\neq y}=\{(x,y)\,|\, x \neq y\}\subset \mathbb{R}^2
$.
\end{example}

In order to prove some results of this work, we have derived the domains of all the Lie algebras of the GKO classification. Since this is a rather trivial calculation, we do not describe it here and we just detail our results in Table \ref{table1}.

\begin{definition}
A finite-dimensional real Lie algebra $V$  of vector fields on an open subset $U\subset \mathbb{R}^2$ is {\it imprimitive} when there exists a one-dimensional distribution $\mathcal{D}$ on $\mathbb{R}^2$ invariant under the action of $V$ by Lie brackets, i.e.  for every  $X\in V$ and every vector field $Y$ taking values in $\mathcal{D}$, we have that $[X,Y]$ takes values in $\mathcal{D}$. Otherwise, $V$ is called   {\em primitive}.
\end{definition}

\begin{example}
Recall that ${\rm I}_4$ is spanned by the vector fields $X_1,X_2$ and $X_3$ given in Table \ref{table1}. If we define $\mathcal{D}$ to be the distribution  on $\mathbb{R}^2$ generated
by $Y=\partial_x$, we see that
$$
[X_1,Y]=0,\qquad [X_2,Y]=-Y,\qquad [X_3,Y]=-2xY.
$$
We infer from this that $\mathcal{D}$ is a one-dimensional distribution invariant under the action of I$_4$. Hence, I$_4$ is an imprimitive Lie algebra of vector fields.
\end{example}

Apart from this first division into primitive/imprimitive Lie algebras, GKO subdivided the primitive finite-dimensional Lie algebras into eight families (P$_i$) and the imprimitive ones into twenty classes (I$_i$). Notice that several of  them  depend on some parameters (such as P$_1$, I$_8$ and I$_{16}$) and that the same Lie algebra structure may appear several times, e.g. I$_3$--I$_5$ and I$_6$, I$_7$, although such Lie algebras are not locally diffeomorphic among themselves. Some of the Lie algebras of Table \ref{table1} can be considered as Lie subalgebras of other classes, e.g. ${\rm P}_6$ is a Lie subalgebra of ${\rm P}_8$. A non-exhaustive list of relations of inclusion among elements of the different Lie algebras of the GKO classification is displayed in Table \ref{table2}.
This list fulfils many details not given in \cite{GKP92} and, if we study a Lie algebra of vector fields which does not consists of Hamiltonian vector fields, we can use them to find which of their subalgebras do. 

For our further purposes, we stress that the class I$_{14}$  contains Lie algebras which are not isomorphic  depending on the choice of the functions $\eta_j$. For instance, if we take $r=1$ and $\eta_1(x)=1$, then we have an instance of  I$_{14}^{r=1}$ given by $\langle X_1= \partial_x, X_2=\partial_y\rangle\simeq \mathbb R^2$. Meanwhile, if we set $r=1$ and $\eta_1(x)={\rm e}^x$, then we get the Lie algebra of I$_{14}^{r=1}$ of the form $\langle X_1= \partial_x, X_2={\rm e}^x\partial_y\rangle\simeq {\mathfrak h}_2$. Notice that a similar fact also appears  within I$_{15}$.

\begin{table}[t] {\footnotesize
 \noindent
\caption{{\small The GKO classification of  the $8+20$  finite-dimensional real Lie algebras of vector fields on the plane and their most relevant
characteristics. The first (one or two) vector fields which are written between brackets form a modular generating system.
The functions $\xi_1(x),\ldots,\xi_r(x)$ and $1$ are linearly independent and the functions $\eta_1(x),\ldots,\eta_r(x)$ form a basis of solutions for a system of $r$ linear differential equations in normal form with constant coefficients \cite[pp.~470--471]{HA75}. Finally, $\mathfrak{g}=\mathfrak{g}_1\ltimes \mathfrak{g}_2$ means that $\mathfrak{g}$ is the direct sum (as linear subspaces) of $\mathfrak{g}_1$ and $\mathfrak{g}_2$, with $\mathfrak{g}_2$ being an ideal of $\mathfrak{g}$.}}
\label{table1}
\medskip
\noindent\hfill
 \begin{tabular}{ p{.5cm} p{2.3cm}    p{10.5cm} l}
\hline
&  &\\[-1.9ex]
\#&Primitive & Basis of vector fields $X_i$ &Domain\\[+1.0ex]
\hline
 &  &\\[-1.9ex]
P$_1$&$A_\alpha\simeq \mathbb{R}\ltimes \mathbb{R}^2$ & $ \{ { {\partial_x} ,    {\partial_y} \},   \alpha(x\partial_x + y\partial_y)  +  y\partial_x - x\partial_y},\quad \ \alpha\geq 0$&$\mathbb{R}^2$ \\[+1.0ex]
P$_2$&$\mathfrak{sl}(2)$ & $\{ {\partial_x},   {x\partial_x  +  y\partial_y} \},   (x^2  -  y^2)\partial_x  +  2xy\partial_y$&$\mathbb{R}^2_{y\neq 0}$\\[+1.0ex]
P$_3$&$\mathfrak{so}(3)$ &${     \{{ y\partial_x  -  x\partial _y},     { (1  +  x^2  -  y^2)\partial_x  +  2xy\partial_y} \},   2xy\partial_x  +  (1  +  y^2  -  x^2)\partial_y}$&$\mathbb{R}^2$\\[+1.0ex]
P$_4$&$\mathbb{R}^2\ltimes\mathbb{R}^2$ &$ \{ {\partial_x},   {\partial_y}\},  x\partial_x + y\partial_y,   y\partial_x - x\partial_y$&$\mathbb{R}^2$\\[+1.0ex]
P$_5$&$\mathfrak{sl}(2 )\ltimes\mathbb{R}^2$ &${ \{ {\partial_x},   {\partial_y}\},  x\partial_x - y\partial_y,  y\partial_x,  x\partial_y}$&$\mathbb{R}^2$\\[+1.0ex]
P$_6$&$\mathfrak{gl}(2 )\ltimes\mathbb{R}^2$ &${ \{{\partial_x},    {\partial_y}\},   x\partial_x,   y\partial_x,   x\partial_y,   y\partial_y}$&$\mathbb{R}^2$\\[+1.0ex]
P$_7$&$\mathfrak{so}(3,1)$ &${ \{ {\partial_x},   {\partial_y}\},   x\partial_x\!+\! y\partial_y,   y\partial_x \!-\! x\partial_y,   (x^2 \!-\! y^2)\partial_x \!+\! 2xy\partial_y,  2xy\partial_x \!+\! (y^2\!-\!x^2)\partial_y}$  &$\mathbb{R}^2$\\[+1.0ex]
P$_8$&$\mathfrak{sl}(3 )$ &${ \{ {\partial_x},    {\partial_y}\},   x\partial_x,   y\partial_x,   x\partial_y,   y\partial_y,   x^2\partial_x + xy\partial_y,   xy\partial_x  +  y^2\partial_y}$&$\mathbb{R}^2$\\[+1.5ex]
\hline
&  &\\[-1.5ex]
\#& Imprimitive\!\! & Basis of vector fields $X_i$ &Domain\\[+1.0ex]
\hline
&  &\\[-1.5ex]
I$_1$&$\mathbb{R}$ &$ \{ {\partial_x} \}$ & $\mathbb{R}^2$\\[+1.0ex]
I$_2$&$\mathfrak{h}_2$ & $ \{ {\partial_x} \},  x\partial_x$& $\mathbb{R}^2$\\[+1.0ex]
I$_3$&$\mathfrak{sl}(2 )$ (type I) &$ \{ {\partial_x}\},  x\partial_x,  x^2\partial_x$& $\mathbb{R}^2$\\[+1.0ex]
I$_4$&$\mathfrak{sl}(2 )$ (type II) & ${ \{ {\partial_x  +  \partial_y},    {x\partial _x + y\partial_y}\},   x^2\partial_x  +  y^2\partial_y}$ &$\mathbb{R}^2_{x\neq y}$\\[+1.0ex]
I$_5$&$\mathfrak{sl}(2 )$ (type III) &${\{ {\partial_x},    {2x\partial_x + y\partial_y}\},   x ^2\partial_x  +  xy\partial_y}$&$\mathbb{R}^2_{y\neq 0}$\\[+1.0ex]
I$_6$&$\mathfrak{gl}(2 )$ (type I)& ${\{ {\partial_x},    {\partial_y}\},   x\partial_x,   x^2\partial_x}$&$\mathbb{R}^2$\\[+1.0ex]
I$_7$&$\mathfrak{gl}(2 )$ (type II)& ${ \{ {\partial_x},   {y\partial_y} \},     x\partial_x,    x^2\partial_x +  xy \partial_y}$&$\mathbb{R}^2_{y\neq 0}$ \\[+1.0ex]
I$_8$&$B_\alpha\simeq \mathbb{R}\ltimes\mathbb{R}^2$ &${ \{ {\partial_x},    {\partial_y}\},   x\partial_x  +  \alpha y\partial_y},\quad  0<|\alpha|\leq 1$&$\mathbb{R}^2$\\[+1.0ex]
I$_9$&$\mathfrak{h}_2\oplus\mathfrak{h}_2$ &${\{ {\partial_x},    {\partial_y}\},   x\partial_x,  y\partial_y}$&$\mathbb{R}^2$\\[+1.0ex]
I$_{10}$&$\mathfrak{sl}(2 )\oplus \mathfrak{h}_2$ & ${\{ {\partial_x},    {\partial_y} \},   x\partial_x,  y\partial_y,  x^2\partial_x }$&$\mathbb{R}^2$\\[+1.0ex]
I$_{11}$&$\mathfrak{sl}(2 )\oplus\mathfrak{sl}(2 )$ &$ \{ {\partial_x},    {\partial_y}\},   x\partial_x,   y\partial_y,   x^2\partial_x ,   y^2\partial_y $&$\mathbb{R}^2$\\[+1.0ex]
I$_{12}$&$\mathbb{R}^{r + 1}$ &$\{ {\partial_y} \},   \xi_1(x)\partial_y, \ldots , \xi_r(x)\partial_y,\quad   r\geq 1$&$\mathbb{R}^2$ \\[+1.0ex]
I$_{13}$&$\mathbb{R}\ltimes \mathbb{R}^{r + 1}$ &$ \{ {\partial_y} \},   y\partial_y,    \xi_1(x)\partial_y, \ldots , \xi_r(x)\partial_y,\quad   r\geq 1$ &$\mathbb{R}^2$\\[+1.0ex]
I$_{14}$&$\mathbb{R}\ltimes \mathbb{R}^{r}$ & ${\{ {\partial_x},   {\eta_1(x)\partial_y} \},  {\eta_2(x)\partial_y},\ldots ,\eta_r(x)\partial_y},\quad r\geq 1$&$\mathbb{R}^2$\\[+1.0ex]
I$_{15}$&$\mathbb{R}^2\ltimes \mathbb{R}^{r}$ &  ${\{ {\partial_x},    {y\partial_y}\} ,    {\eta_1(x)\partial_y},\ldots, \eta_r(x)\partial_y},\quad  r\geq 1$&$\mathbb{R}^2$\\[+1.0ex]
I$_{16}$&$C_\alpha^r\simeq \mathfrak{h}_2\ltimes\mathbb{R}^{r + 1}$ & ${ \{ {\partial_x},    {\partial_y} \},   x\partial_x  +  \alpha y\partial y,   x\partial_y, \ldots, x^r\partial_y},\quad   r\geq 1,\qquad \alpha\in\mathbb{R}$&$\mathbb{R}^2$\\[+1.0ex]
I$_{17}$&$\mathbb{R}\ltimes(\mathbb{R}\ltimes \mathbb{R}^{r})$ &$ \{ {\partial_x},    {\partial_y} \},   x\partial_x  +  (ry  +  x^r)\partial_y ,   x\partial_y, \ldots,  x^{r - 1}\partial_y,\quad   r\geq 1$ &$\mathbb{R}^2$\\[+1.0ex]
I$_{18}$&$(\mathfrak{h}_2\oplus \mathbb{R})\ltimes \mathbb{R}^{r + 1}$ & $ \{ {\partial_x},    {\partial_y}\},   x\partial_x,   x\partial_y,   y\partial_y,   x^2\partial_y, \ldots,x^r\partial_y,\quad r\geq 1$ &$\mathbb{R}^2$\\[+1.0ex]
I$_{19}$&$\mathfrak{sl}(2 )\ltimes \mathbb{R}^{r + 1}$ &  $ \{ {\partial_x},    {\partial_y} \},   x\partial_y,    2x\partial _x  +  ry\partial_y,   x^2\partial_x  +  rxy\partial_y,   x^2\partial_y,\ldots, x^r\partial_y ,\quad   r\geq 1$&$\mathbb{R}^2$ \\[+1.0ex]
I$_{20}$&$\mathfrak{gl}(2 )\ltimes \mathbb{R}^{r + 1}$ &  $ \{ {\partial_x},    {\partial_y} \},   x\partial_x,   x\partial_y,   y\partial_y,   x^2\partial_x  +  rxy\partial_y,   x^2\partial_y,\ldots, x^r\partial_y,\quad   r\geq 1$ &$\mathbb{R}^2$\\[+1.5ex]
\hline
 \end{tabular}
\hfill}
\end{table}


\begin{table}[t] {\footnotesize
 \noindent
\caption{{\small {Tree of inclusion relations} among the Lie algebras of the GKO classification. If (some of) the elements of $A$ are diffeomorphic to Lie subalgebras of $B$, there is a path of (dashed) arrows from $A$ to $B$. As every Lie algebra includes $I_1$, this Lie algebra is not shown. In bold and italics are classes with Hamiltonian Lie algebras and rank one associated distribution, correspondingly. Colors have been employed to help distinguishing arrows.}}
\label{table2}
\medskip
\xymatrix@R=1.8mm@C=3.4mm{
 \dim\,>6             &             &              &{\rm P}_8&                         &&&&&{\rm I}_{20}&&\\
 \dim\,6\rightarrow   &       &          {\rm P}_7 &{\rm P}_6\ar[u]         &          {\rm I}_{11}&&&         & {\rm I}_{19}\ar[ru]  &{\rm I}_{18}\ar[u]&      &\\
\dim\,5\rightarrow&          &           &{\rm \bf P}_5\ar[u]\ar[lu]          &{\rm I}_{10}\ar[u]&&   &   &{\rm \bf I}_{16}\ar[ru] &&{\rm I}_{15}\ar@{--}[ul]& & \\
\dim\,4\rightarrow& &{\rm P}_4\ar[ruu]\ar[uu]&     &{\rm I}_6\ar[u] &{\rm I}_9\ar@/_{1mm}/[rrrruu]\ar[lu]&{\rm I}_7\ar@/^{4mm}/[uuurrr]      &&{\rm I}_{17}\ar[u]&{\color{black} {\it I}_{13}}\ar@{--}[ur]&{\rm I}_{14}^{r> 1}\ar[u]&&\\
\dim\,3\rightarrow  &{\rm \bf P}_3\ar[ruuu]    &{\rm \bf P}_2\ar[ruu]&{\rm \bf P}_1\ar@{--}[uu]\ar[lu]&{\color{black}{\it I}_3}\ar[u]
&{\rm \bf I}_8\ar[u]&{\rm \bf I}_5\ar[u]\ar[uuurr]      &{\rm \bf I}_4\ar@/_{4mm}/[llluuu]&&{\color{black}{\it \bf I}_{12}^{r>1}}\ar@{--}[ul]\ar@{--}[ur]\ar[u]&&&\\
&&&&&&&&&& &\\
\dim\,2\rightarrow          &          &&  {\color{brown}{\bf I}^{r=1}_{14A}\simeq \mathfrak{h}_2}\ar@[brown]@/^{2mm}/[rrruu]\ar@[brown]@/_{2mm}/[rrrrruuu]\ar@[brown]@/^{2mm}/[rrrruu]\ar@[brown][luu]\ar@[brown]@/_{2mm}/[uuul]               &{\color{orange}{\it I}_2\simeq \mathfrak{h}_2}\ar@[orange]@/_{4mm}/[rrrruuu]\ar@[orange]@/_{12mm}/[rrrrrruuuu]\ar@[orange]@/_{4mm}/[rrrrruuu]\ar@[orange]@/_{3mm}/[luuuuu]\ar@[orange][uu]\ar@[orange][uuur]\ar@[orange][uuurr]&&{\color{blue}{\rm \bf I}^{r=1}_{14B}\simeq \mathbb{R}^2\ar@[blue]@/_{4mm}/[uuu]\ar@[blue][llluuuu]\ar@[blue]@/^{2mm}/[llluu]\ar@[blue][rruuuuu]\ar@[blue]@/_{3mm}/[rruuu]\ar@[blue][luu]}& {\color{red}{\it I}_{12}^{r=1}\simeq\mathbb{R}^2}\ar@/^{1mm}/@[red][lllluuuu]\ar@/^{2mm}/@[red][ruuuuu]\ar@[red][ruuu]\ar@[red][rruu]\ar@[red]@/_{7mm}/[rrruuu]&          &                     &&&&\\ 
\save "2,1"."7,8"*\frm<8pt>{.}\restore
\save "1,1"."7,12"**\frm<8pt>{--}\restore
}}
\end{table}

\section{Minimal Lie algebras of Lie--Hamilton systems on the plane}

In this section we study the local structure of the minimal Lie algebras of Lie--Hamilton systems on the plane around their generic points. Our main result, Theorem \ref{Char}, and the remaining findings of this section enable us  to give the local classification of Lie--Hamilton systems on the plane in Section \ref{Classification}.  To simplify the notation, $U$ will hereafter stand for a contractible open subset of $\mathbb{R}^2$.

\begin{lemma}\label{lem:local_sym} Let $V$ be a finite-dimensional real Lie algebra of Hamiltonian vector fields on $\mathbb{R}^2$ with respect to a Poisson structure and let $\xi_0\in\mathbb{R}^2$ be a generic point of $V$. There exists a $U\ni\xi_0$ such that $V|_{U}$ consists of Hamiltonian vector fields relative to a symplectic structure.
\end{lemma}
\begin{proof}

If $\dim \mathcal{D}^V_{\xi_0}=0$, then $\dim \mathcal{D}_\xi^V=0$ for every $\xi$ in a $U\ni \xi_0$ because the rank of $\mathcal{D}^V$ is locally constant around generic points. Consequently, $V|_U=0$  and its unique element become Hamiltonian relative to the restriction of
$\omega=\dd x\wedge \dd y$ to $U$.
Let us assume now $\dim \mathcal{D}^V_{\xi_0}\neq 0$. By assumption, the elements of $V$ are Hamiltonian vector fields with respect to a Poisson bivector $\Lambda\in \Gamma(\Lambda^2{\rm T}\mathbb{R}^2)$.
Hence, $\mathcal{D}^V_\xi\subset \mathcal{D}^\Lambda_\xi$ for every $\xi\in\mathbb{R}^2$, with $\mathcal{D}^\Lambda$ being the {\it characteristic distribution} of $\Lambda$ \cite{IV}. Since $\dim \mathcal{D}_{\xi_0}^V\neq 0$ and $r^V$ is locally constant at $\xi_0$, then $\dim\,\mathcal{D}^V_\xi>0$ for every $\xi$ in a $U\ni \xi_0$. Since the rank of $\mathcal{D}^\Lambda$ is even at every point of $\mathbb{R}^2$ and $\mathcal{D}^V_\xi\subset \mathcal{D}^\Lambda_\xi$ for every $\xi\in U$, the rank of $\mathcal{D}^\Lambda$ is two on $U$. So, $\Lambda$ comes from a symplectic structure on $U$ and $V|_U$ is a Lie algebra of Hamiltonian vector fields relative to it.
\end{proof}

Roughly speaking, the previous lemma establishes that any Lie--Hamilton system $X$ on $\mathbb{R}^2$ can be considered around each generic point of ${V^X}$ as a Lie--Hamilton system admitting a minimal Lie algebra of Hamiltonian vector fields with respect to a symplectic structure. As our study of such systems is local, we hereafter focus on analysing  minimal Lie algebras of this type.

A {\it volume form} $\Omega$ on an $n$-dimensional manifold $N$ is a non-vanishing $n$-form on $N$. The divergence of a vector field $X$ on $N$ with respect to $\Omega$ is the unique function ${\rm div} X:N\rightarrow\mathbb{R}$ satisfying $\mathcal{L}_X\Omega=({\rm div}\, X)\, \Omega$. An {\it integrating factor} for $X$ on $U\subset N$ is a function $f:U\rightarrow \mathbb{R}$ such that $\mathcal{L}_{fX}\Omega=0$ on $U$. Next we have the following result~\cite{LP12}.

\begin{lemma}\label{Lem:IntFac}  Consider the volume form $\Omega=\dd x\wedge \dd y$ on a $U\subset \mathbb{R}^2$ and a vector field $X$ on $U$. Then, $X$ is
Hamiltonian with respect to a symplectic form $\omega=f \Omega$ on $U$  if and only if $f:U\rightarrow \mathbb{R}$ is a non-vanishing integrating factor of $X$ with respect to $\Omega$, i.e. $Xf=-f{\rm div} X$ on $U$.
\end{lemma}
\begin{proof} Since $\omega$ is a symplectic form on $U$, then $f$ must be non-vanishing. As
$$
\mathcal{L}_X\omega=\mathcal{L}_X(f\Omega)=(Xf)\Omega+f\mathcal{L}_X\Omega=(Xf+f{\rm div} X)\Omega=\mathcal{L}_{fX}\Omega,
$$
then $X$ is locally Hamiltonian with respect to $\omega$, i.e. $\mathcal{L}_X\omega=0$, if and only if $f$ is a non-vanishing integrating factor for $X$ on $U$. As $U$ is a contractible open {\ subset}, the Poincar\'e Lemma ensures that $X$ is a local Hamiltonian vector field if and only if it is a Hamiltonian vector field. Consequently, the lemma follows.
\end{proof}

\begin{definition} Given a vector space $V$ of vector fields on $U$, we say that $V$ admits a {\it modular generating system} $(U_1,X_1,\ldots,X_p)$ if $U_1$ is a dense open subset of $U$ such that every $X\in V|_{U_1}$ can be brought into the form
$
X|_{U_1}=\sum_{i=1}^p\ff_i X_i|_{U_1}
$
for certain functions $\ff_1,\ldots,\ff_p\in C^\infty(U_1)$ and vector fields $X_1,\ldots,X_p\in V$.
\end{definition}

\begin{example}
Given the Lie algebra P$_3\simeq \mathfrak {so}(3)$ on $\mathbb{R}^2$ of Table~\ref{table1}, the vector fields
$$
X_1=y\frac{\partial}{\partial x}-x\frac{\partial}{\partial y}  ,\qquad X_2=(1+x^2-y^2)\frac{\partial}{\partial x}+2xy\frac{\partial }{\partial y}
$$
of P$_3$ satisfy that $X_3=\ff_1X_1+\ff_2X_2$ on $U_1=\{(x,y)\in \mathbb{R}^2\mid x\neq 0\}$ for the functions $\ff_1,\ff_2\in C^\infty(U_1)$:
$$
\ff_1=\frac{x^2+y^2-1}{x},\qquad \ff_2=\frac y {x}.
$$
 Obviously, $U_1$ is an open dense subset of $\mathbb{R}^2$. As every element of $V$ is a linear combination of $X_1,X_2$ and $X_3=\ff_1X_1+\ff_2X_2$, then  every $X\in V|_{U_1}$ can be written as a linear combination with smooth functions on $U_1$ of $X_1$ and $X_2$. So, $(U_1,X_1,X_2)$ form a generating modular system for P$_3$.
\end{example}

In Table \ref{table1} we detail a modular generating system, which is indicated by the first one or two vector fields written between brackets in the list of the $X_i$'s, for every finite-dimensional Lie algebra of vector fields of the GKO classification.

\begin{theorem}\label{Char} Let $V$ be a Lie algebra of vector fields on $U\subset \mathbb{R}^2$ admitting a modular generating system $(U_1,X_1,\ldots,X_p)$.  We have that:

\noindent
1) The space $V$ consists of Hamiltonian vector fields relative to a symplectic form on $U$ if and only if:
\begin{enumerate}[i)]
 \item Let $\ff_1,\ldots,\ff_p$ be certain smooth functions on $U_1\subset U$. Then,
 \begin{equation}\label{DivCon}
 X|_{U_1}=\sum_{i=1}^p \ff_i X_i|_{U_1}\in V|_{U_1}\Longrightarrow {\rm div} X|_{U_1}=\sum_{i=1}^p \ff_i {\rm div}X_i|_{U_1}.
 \end{equation}
\item The elements $X_1,\ldots,X_p$ admit a common non-vanishing integrating factor on $U$.
\end{enumerate}

\noindent
2) If the rank of $\mathcal{D}^V$ is two on $U$, the symplectic form is unique up to a multiplicative non-zero constant.

\end{theorem}
\begin{proof} Let us prove the direct part of 1). Since $(U_1,X_1,\ldots,X_p)$ form a modular generating system for $V$, we have that every $X|_{U_1}\in V|_{U_1}$ can be brought into  the form $X|_{U_1}=\sum_{i=1}^p\ff_i X_i|_{U_1}$ for certain $\ff_1,\ldots,\ff_p\in C^\infty(U_1)$. As $V$ is a Lie algebra of Hamiltonian vector fields with respect to a symplectic structure on $U$, let us say
 \begin{equation}\label{ww}
\omega=f(x,y)\dd x\wedge \dd y,
 \end{equation}
then Lemma \ref{Lem:IntFac} ensures that  $Yf=-f{\rm div} Y$ for every $Y\in V$.  Then,
$$
f{\rm div} X=-X f=-\sum_{i=1}^p \ff_i X_i f=f \sum_{i=1}^p \ff_i {\rm div}X_i \Longleftrightarrow
f\left({\rm div} X-\sum_{i=1}^p \ff_i {\rm div}X_i\right)=0
$$
on $U_1$.  
As $\omega$ is non-degenerate, then $f$ is non-vanishing and $i)$ follows. Since all the vector fields of $V$ are Hamiltonian with respect to $\omega$, they share a common non-vanishing integrating factor, namely $f$. From this, $ii)$ holds.

Conversely, if $ii)$ is satisfied, then Lemma \ref{Lem:IntFac} ensures that $X_1,\ldots,X_p$ are Hamiltonian with respect to  (\ref{ww}) on $U$, with $f$ being a non-vanishing integrating factor. As $(U_1,X_1,\ldots,X_p)$ form a generating modular system for $V$, every $X\in V$ can be written as $\sum_{i=1}^p\ff_i X_i$ on $U_1$ for certain functions $\ff_1,\ldots,\ff_p\in C^\infty(U_1)$. From $i)$ we obtain
$
{\rm div}\, X=\sum_{i=1}^p\ff_i{\rm div}X_i
$
on $U_1$. Then,
$$
Xf=\sum_{i=1}^p\ff_i X_i f=-f\sum_{i=1}^p\ff_i{\rm div}X_i=-f{\rm div X}
$$
on $U_1$ and, since the elements of $V$ are smooth and $U_1$ is dense on $U$, the above expresion also holds  on $U$. Hence, $f$ is a non-vanishing integrating factor for $X$, which becomes a Hamiltonian vector field with respect to $\omega$ on $U$ in virtue of Lemma \ref{Lem:IntFac}. Hence, part 1) is proven.

As far as part 2) of the theorem is concerned, if the vector fields of $V$ are Hamiltonian with respect to two different symplectic structures on $U$, they admit two different non-vanishing integrating factors $f_1$ and $f_2$. Hence,
$$X(f_1/f_2)=(f_2Xf_1-f_1Xf_2)/f_2^2=(f_2f_1{\rm div} X-f_1f_2{\rm div}X)/f_2^2=0
$$
and $f_1/f_2$ is a common constant of motion for all the elements of $V$. Hence, it is a constant of motion for every vector field taking values in the distribution $\mathcal{D}^V$. Then rank of $\mathcal{D}^V$ on $U$ is two by assumption. So, $\mathcal{D}^V$ is generated by the vector fields $\partial_x$ and $\partial_y$ on $U$. Thus, the only constants of motion on $U$
common to all the vector fields taking values in $\mathcal{D}^V$, and consequently common to the elements of $V$, are constants. Since $f_1$ and $f_2$ are non-vanishing, then $f_1=\lambda f_2$ for a $\lambda\in\mathbb{R}\backslash \{0\}$ and the associated symplectic structures are the same up to an irrelevant non-zero  proportionality constant.
\end{proof}

Using Theorem \ref{Char}, we can immediately prove the following result.

\begin{corollary}\label{Cor:NonDiv}
If a Lie algebra of vector fields $V$ on a $U\subset\mathbb{R}^2$ consists of Hamiltonian vector fields with respect to a symplectic form and admits a modular generating system whose elements are divergence free, then every element of $V$ is divergence free.
\end{corollary}

\section{Lie--Hamilton algebras}

In this section we prove some new results concerning Lie--Hamilton algebras. Analogues of the following results can also be proved through Lie algebra cohomology techniques \cite{OR04}, although the approach here presented is simpler and provides all the tools that we will need in the following sections.

It is known that Lie--Hamilton algebras are not uniquely defined in general. Moreover, the existence of different types of Lie--Hamilton algebras for the same Lie--Hamilton system is important in their linearisation and the use of certain methods~\cite{CLS12Ham}. For instance, if a Lie--Hamilton system $X$ on $N$ admits a Lie--Hamilton algebra isomorphic to $V^X$ and $\dim V^X=\dim N$, then $X$ can be linearized together with its associated Poisson structure \cite{CLS12Ham}.

\begin{example}
Consider again the Lie--Hamilton system $X$ given by (\ref{Riccati2}) and assume $V^X\simeq \mathfrak{sl}(2)$. Recall that $X$ admits a Lie--Hamilton algebra $(\mathcal{H}_\Lambda,\{\cdot,\cdot\}_\omega)\simeq \mathfrak{sl}(2)$ spanned by the Hamiltonian functions $h_1,h_2,h_3$ given by (\ref{ab}) relative to the symplectic structure $\omega$ detailed in (\ref{aa}). We can also construct a second (non-isomorphic) Lie--Hamilton algebra for $X$ with respect to (\ref{aa}). The vector fields $X_i$, with $i=1,2,3$, spanning $V^X$ (see (\ref{vectRiccati2})) have also Hamiltonian functions $\bar h_i=h_i+1$, for $i=1,2,3$, respectively. Hence, $(\mathbb{R}^2_{y\neq 0},\omega,h=a_0(t)\bar{h}_1+a_1(t)\bar{h}_2+a_2(t)\bar{h}_3)$ is a Lie--Hamiltonian structure for $X$
giving rise to a Lie--Hamilton algebra  $(\overline{\mathcal{H}}_\Lambda,\{\cdot,\cdot\}_\omega)\equiv (\langle \bar{h}_1,\bar{h}_2,\bar{h}_3,1\rangle,\{\cdot,\cdot\}_\omega) \simeq \mathfrak{sl}(2)\oplus \mathbb{R}$ for $X$.
\end{example}

\begin{proposition}\label{LieHam1} A Lie--Hamilton system $X$ on a symplectic connected manifold $(N,\omega)$ possesses an associated Lie--Hamilton algebra $(\mathcal{H}_\Lambda,\{\cdot,\cdot\}_\omega)$ isomorphic to $V^X$ if and only if every Lie--Hamilton algebra non-isomorphic to $V^X$ is isomorphic to $V^X\oplus \mathbb{R}$.
\end{proposition}
\begin{proof} Let $(\overline{\mathcal{H}}_\Lambda,\{\cdot,\cdot\}_\omega)$ be an arbitrary Lie--Hamilton algebra for $X$.
As $X$ is defined on a connected manifold, the sequence of Lie algebras
\begin{equation}\label{ExacSeq}
0\hookrightarrow (\overline{\mathcal{H}}_\Lambda,\{\cdot,\cdot\}_\omega)\cap \langle 1\rangle \hookrightarrow (\overline{\mathcal{H}}_\Lambda,\{\cdot,\cdot\}_\omega)\stackrel{\varphi}{\longrightarrow} V^X\rightarrow 0,
\end{equation}
where $\varphi:\overline{\mathcal{H}}_\Lambda\rightarrow V^X$ maps every function of $\overline{\mathcal{H}}_\Lambda$ to minus its Hamiltonian vector field, is always exact (cf.~\cite{CLS12Ham}). Hence, $(\overline{\mathcal{H}}_\Lambda,\{\cdot,\cdot\}_\omega)$ can be isomorphic either to $V^X$ or to   a Lie algebra extension of $V^X$ of  dimension $\dim V^X+1$. 

If $(\mathcal{H}_\Lambda,\{\cdot,\cdot\}_\omega)$ is isomorphic to $V^X$ and there exists a second Lie--Hamilton algebra $(\overline{\mathcal{H}}_\Lambda,\{\cdot,\cdot\}_\omega)$ for $X$ non-isomorphic to $V^X$, we see from (\ref{ExacSeq}) that $1\in \overline{\mathcal{H}}_\Lambda$ and $1\notin \mathcal{H}_\Lambda$. Given a basis $X_1,\ldots,X_r$ of $V^X$, each element $X_i$, with $i=1,\ldots,r$, has a Hamiltonian function $\overline{h}_i\in\overline{\mathcal{H}}_\Lambda$ and another $h_i\in{\mathcal{H}}_\Lambda$. As $V^X$ is defined on a connected manifold, then $h_i=\overline{h}_i-\lambda_i\in \overline{\mathcal{H}}_\Lambda$ with $\lambda_i\in\mathbb{R}$ for every $i=1,\ldots,r$. From this and using again that $1\in \overline{\mathcal{H}}_\Lambda\backslash \mathcal{H}_\Lambda$, we obtain that $\{h_1,\ldots,h_r,1\}$ is a basis for $\overline{\mathcal{H}}_\Lambda$ and $(\overline{\mathcal{H}}_\Lambda,\{\cdot,\cdot\}_\omega)\simeq (\mathcal{H}_\Lambda \oplus\mathbb{R},\{ \cdot,\cdot\}_\omega)$.

Let us assume now that every Lie--Hamilton algebra $(\overline{\mathcal{H}}_\Lambda,\{\cdot,\cdot\}_\omega)$ non-isomorphic to $V^X$ is isomorphic to $V^X\oplus\mathbb{R}$. We can define a Lie algebra anti-isomorphism $\mu :V^X\rightarrow \overline{\mathcal{H}}_\Lambda$ mapping each element of $V^X$ to a Hamiltonian function belonging to a Lie subalgebra of $(\overline{\mathcal{H}}_\Lambda,\{\cdot,\cdot\}_\omega)$ isomorphic to $V^X$. Hence, $(N,\omega,h=\mu(X))$, where $h_t=\mu(X_t)$ for each $t\in\mathbb{R}$, is a Lie--Hamiltonian structure for $X$ and $(\mu(V^X),\{\cdot,\cdot\}_\omega)$ is a Lie--Hamilton algebra for $X$ isomorphic to $V^X$.
\end{proof}

\begin{proposition}\label{LieHam2} If a Lie--Hamilton system $X$ on a symplectic connected manifold $(N,\omega)$ admits an associated Lie--Hamilton algebra $(\mathcal{H}_\Lambda,\{\cdot,\cdot\}_\omega)$ isomorphic to $V^X$, then it admits a Lie--Hamilton algebra isomorphic to $V^X\oplus \mathbb{R}$.
\end{proposition}

\begin{proof}
Let $(N,\omega,h)$ be a Lie--Hamiltonian structure for $X$ giving rise to the Lie--Hamilton algebra $(\mathcal{H}_\Lambda,\{\cdot,\cdot\}_\omega)$. Consider the linear space $L_h$ spanned by linear combinations of the functions $\{h_t\}_{t\in\mathbb{R}}$. Since we assume $\mathcal{H}_\Lambda\simeq V^X$, the exact sequence (\ref{ExacSeq}) involves that $1\notin L_h$. Moreover, we can write $h=\sum_{i=1}^pb_i(t)h_{t_i}$, where $h_{t_i}$ are the values of $h$ at certain times $t_1,\ldots,t_p$ such that $\{h_{t_1},\ldots,h_{t_p}\}$ are linearly independent and $b_1,\ldots,b_p$ are certain $t$-dependent functions. Observe that the vector fields $(b_1(t),\ldots,b_p(t))$, with $t\in\mathbb{R}$, span a $p$-dimensional linear space. If we choose a $t$-dependent Hamiltonian $\bar{h}=\sum_{i=1}^{p}b_i(t)h_{t_i}+b_{p+1}(t)$, where $b_{p+1}(t)$ is not a linear combination of $b_1(t),\ldots,b_p(t)$, and we recall that $1,h_{t_1},\ldots,h_{t_p}$ are linearly independent over $\mathbb{R}$,
we obtain that the linear hull of the functions $\{\bar h_t\}_{t\in\mathbb{R}}$ has dimension $\dim L_h+1$. Moreover, $(N,\{\cdot,\cdot\}_\omega,\bar h)$ is a Lie--Hamiltonian structure for $X$. Hence, they span, along with their successive Lie brackets, a Lie--Hamilton algebra isomorphic to $\mathcal{H}_\Lambda\oplus \mathbb{R}$.
\end{proof}

\begin{corollary}\label{NoGo2} If $X$ is a Lie--Hamilton system with respect to a symplectic connected manifold $(N,\omega)$ admitting a Lie--Hamilton algebra $(\mathcal{H}_\Lambda,\{\cdot,\cdot\}_\omega)$ satisfying that $1\in \{\mathcal{H}_\Lambda,\mathcal{H}_\Lambda\}_\omega$, then $X$ does not possess any Lie--Hamilton algebra isomorphic to $V^X$.
\end{corollary}
\begin{proof} If $1\in \{\mathcal{H}_\Lambda,\mathcal{H}_\Lambda\}_\omega$, then $\mathcal{H}_\Lambda$ cannot be isomorphic to $V^X\oplus\mathbb{R}$ because the derived Lie algebra of $\mathcal{H}_\Lambda$, i.e. $\{\mathcal{H}_\Lambda,\mathcal{H}_\Lambda\}_\omega$, contains the constant function $1$ and the derived Lie algebra of a $\mathcal{H}_\Lambda$ isomorphic to $V^X\oplus \mathbb{R}$ does not. In view of Proposition \ref{LieHam1}, system $X$ does not admit any Lie--Hamilton algebra isomorphic to $V^X$.
\end{proof}

\begin{proposition}\label{UniExt} If $X$ is a Lie--Hamilton system on a connected manifold $N$ admitting a $V^X$ of Hamiltonian vector fields with respect to a symplectic structure $\omega$ that does not possess any Lie--Hamilton algebra $(\mathcal{H}_\Lambda,\{\cdot,\cdot\}_\omega)$ isomorphic to $V^X$, then all its Lie--Hamilton algebras (with respect to the Lie bracket $\{\cdot,\cdot\}_\omega$) are {\b isomorphic}.
\end{proposition}
\begin{proof} Let $(\mathcal{H}_\Lambda,\{\cdot,\cdot\}_\omega)$ and $(\overline{\mathcal{H}}_\Lambda,\{\cdot,\cdot\}_\omega)$ be two Lie--Hamilton algebras for $X$. Since they are not isomorphic to $V^X$ and in view of the exact sequence (\ref{ExacSeq}), then $1\in \mathcal{H}_\Lambda\cap \overline{\mathcal{H}}_\Lambda$ . Let $X_1,\ldots,X_r$ be a basis of $V^X$. Every vector field $X_i$ admits a Hamiltonian function $h_i\in\mathcal{H}_\Lambda$ and another $\bar{h}_i\in\overline{\mathcal{H}}_\Lambda$. The functions $h_1,\ldots,h_r$ are linearly independent and the same applies to $\bar h_1,\ldots,\bar h_r$. Then,  $\{h_1,\ldots,h_r,1\}$ is a basis for $\mathcal{H}_\Lambda$ and  $\{\bar{h}_1,\ldots,\bar{h}_r,1\}$   is a basis for $\overline{\mathcal{H}}_\Lambda$. As $N$ is connected, then $h_i=\bar{h}_i -\lambda_i$ with $\lambda_i\in\mathbb{R}$ for each $i\in\mathbb{R}$. Hence, the functions $h_i$ belong to $\overline{\mathcal{H}}_\Lambda$ and the functions $\bar{h}_i$ belong to ${\mathcal{H}}_\Lambda$. Thus $\mathcal{H}_\Lambda=\overline{\mathcal{H}}_\Lambda$.
\end{proof}

\section{Local classification of Lie--Hamilton systems on the plane}\label{Classification}

In this section we describe the local structure of Lie--Hamilton systems on the plane, i.e. given the minimal Lie algebra of a Lie--Hamilton system $X$ on the plane, we prove that $V^X$ is locally diffeomorphic around a generic point of $V^X$ to one of the Lie algebras given in Table \ref{table3}. We also prove that, around a generic point of $V^X$, the Lie--Hamilton algebras of $X$ must have one of the algebraic structures described in Table \ref{table3}.

If $X$ is a Lie--Hamilton system, its minimal Lie algebra must be locally diffeomorphic to one of the Lie algebras of the GKO classification that consists of Hamiltonian vector fields with respect to a Poisson structure. As we are concerned with generic points of minimal Lie algebras, Lemma \ref{lem:local_sym} ensures that $V^X$ is locally diffeomorphic around generic points to a Lie algebra of Hamiltonian vector fields with respect to a symplectic structure. So, its minimal Lie algebra is locally diffeomorphic to one of the Lie algebras of the GKO classification consisting of Hamiltonian vector fields with respect to a symplectic structure on a certain open contractible subset of its domain. By determining which of the Lie algebras of the GKO classification admit such a property, we can classify the local structure of all Lie--Hamilton systems on the plane. This relevant result can be stated as follows: 

\begin{proposition}\label{NoGo} The primitive Lie algebras ${\rm  P}^{\alpha\ne 0}_1$, {\rm P}$_4$, {\rm P}$_6$--{\rm P}$_8$ and the imprimitive ones {\rm I}$_2$, {\rm I}$_3$, {\rm I}$_6$, {\rm I}$_7$, {\rm I}$^{(\alpha\neq -1)}_8$, {\rm I}$_9$--{\rm I}$_{11}$, {\rm I}$_{13}$, {\rm I}$_{15}$, ${\rm I}^{(\alpha\neq -1)}_{16}$, {\rm I}$_{17}$--{\rm I}$_{20}$ are not Lie algebras of Hamiltonian vector fields on any $U\subset \mathbb{R}^2$.
\end{proposition}

\begin{proof}
Apart from I$_{15}$, the remaining Lie algebras detailed in this statement  admit a modular generating system whose elements are divergence free on the whole $\mathbb{R}^2$ (see the elements between brackets in Table \ref{table1}). At the same time, we also observe in Table \ref{table1} that these Lie algebras admit a vector field with non-zero divergence on any $U$. In view of Corollary \ref{Cor:NonDiv}, they cannot be Lie algebras of Hamiltonian vector fields with respect to any symplectic structure on any $U\subset\mathbb{R}^2$.

In the case of the Lie algebra I$_{15}$, we have that $(\mathbb{R}^2_{y\neq 0},X_1=\partial_x,X_2=y\partial_y)$ form a generating modular system of I$_{15}$. Observe that $X_2=y\partial_y$ and $X_3=\eta_1(x)\partial_y$, where $\eta_1$ is a non-null function ---it forms with $\eta_2(x),\ldots,\eta_r(x)$ a basis of solutions of a system of $r$ first-order linear homogeneous differential equations in normal form with constant coefficients (cf. \cite{GKP92,1880})--- satisfy ${\rm div} X_2=1$ and ${\rm div}X_3=0$. Obviously, ${\rm div} X_3\neq {\rm div}X_2/\eta_1$ on any $U$. So,
I$_{15}$ does not satisfy Theorem \ref{Char} on any $U$ and it is not a Lie algebra of Hamiltonian vector fields on any  $U\subset \mathbb{R}^2$.
\end{proof}

To simplify the notation, we assume in this section that all objects are defined on a contractible $U\subset \mathbb{R}^2$ of the domain of the Lie algebra under study. Additionally, $U_1$ stands for a dense open subset of $U$. In the following two subsections, we explicitly show that {\em all} of  the Lie algebras of the GKO classification {\em not} listed in Proposition \ref{NoGo} consist of Hamiltonian vector fields on any $U$ of their domains. For each Lie algebra, we compute the integrating factor $f$ of $\omega$ given by (\ref{ww}) turning the elements of a basis of the Lie algebra into Hamiltonian vector fields and we work out their Hamiltonian functions.
Additionally, we obtain the algebraic structure of all the Lie--Hamilton algebras of the Lie--Hamilton systems admitting such minimal Lie algebras.

We stress that the main results covering the resulting Hamiltonian functions $h_i$, the symplectic form $\omega$ and the Lie Hamiltonian algebra are summarized in Table~\ref{table3} accordingly to the GKO classification of Table~\ref{table1}, so that the reader may skip all the details given in Subsections 6.1 and 6.2 concerning the corresponding computations for primitive Lie--Hamiltonian algebras and imprimitive ones. In this respect,
we  point out that  the Lie algebras of the class I$_{14}$   gives rise to {\em two} non-isomorphic Lie--Hamilton algebras: I$_{14A}$ whenever $1\notin \langle \eta_1,\ldots,\eta_r\rangle$ and I$_{14B}$ otherwise. Consequently, we  obtain {\em twelve}  finite-dimensional real  Lie algebras of Hamiltonian vector fields on the plane.

In order to shorten the presentation of the following results, we remark that for some of such Lie--Hamilton algebras their symplectic structure is just the standard one:

\begin{proposition}\label{GoSym}
The Lie algebras ${\rm P}^{(\alpha=0)}_1$, {\rm P}$_5$, {\rm I}$_8^{(\alpha=-1)}$,  {\rm I}$_{14B}$ and {\rm I}$_{16}^{(\alpha=-1)}$   are Lie algebras of Hamiltonian vector fields  with respect to the standard symplectic form $\omega=\dd x\wedge \dd y$, that is, $f\equiv  1$.
\end{proposition}
\begin{proof}
We see in Table \ref{table1} that all the aforementioned Lie algebras admit a modular generating system $(U,X_1=\partial_x,X_2=\partial_y)$ and all their elements have zero divergence. So, they satisfy condition (\ref{DivCon}). The vector fields $X_1,X_2$ are Hamiltonian with respect to the symplectic structure $\omega=\dd x\wedge \dd y$. In view of Theorem \ref{Char}, the whole Lie algebra consists of Hamiltonian vector fields relative to $\omega$.
\end{proof}

\begin{table}[t] {\footnotesize
 \noindent
\caption{{\small The   classification of the $4+8$ finite-dimensional real  Lie algebras of Hamiltonian vector fields on $\mathbb{R}^2$. For I$_{12}$, I$_{14A}$ and I$_{16}$, we have $j=1,\dots,r$ and $ r\ge1$; in  I$_{14B}$ the index $j=2,\dots, r$.
  }}
\label{table3}
\medskip
\noindent\hfill
\begin{tabular}{ l l l l l }
\hline
&&&  &\\[-1.5ex]
\#&Primitive  & Hamiltonian functions $h_i$& $\omega$ &  Lie--Hamilton algebra\ \\[+1.0ex]
\hline
&  & &  &\\[-1.5ex]
P$_1$& $A_0\simeq {\mathfrak{iso}}(2)$ & ${y, \ -x, \ \tfrac 12 (x^2+y^2)},\ 1$ & ${\rm d}x\wedge {\rm d}y$ & $ \overline{\mathfrak{iso}}(2)$\\[+1.0ex]
P$_2$& $\mathfrak{sl}(2 )$ & $\displaystyle{- \frac 1y, \ -\frac xy, \  -\frac{x^2+y^2}{y}}$ & $\displaystyle{    \frac{ {\rm d}x\wedge {\rm  d}y }{ y^{2}  } }$ & $\mathfrak{sl}(2)$ or ${\mathfrak{sl}}(2)\oplus\mathbb{R}$\\[+2ex]
P$_3$& $\mathfrak{so}(3)$ &$\displaystyle { \frac{-1}{2 (1+x^2+y^2) },\  \frac{ y}{1+x^2+y^2},  }$ &$\displaystyle \frac{\dd x\wedge \dd y}{(1+x^2+y^2)^{2}}$ & $\mathfrak{so}(3)$ or ${\mathfrak{so}}(3)\oplus\mathbb{R}$\\[+2.5ex]
& &$\displaystyle { - \frac{x}{1+x^2+y^2} }$, \ 1 &  &  \\[+2.5ex]
 P$_5$& $\mathfrak{sl}(2 )\ltimes\mathbb{R}^2$  & ${y,\  -x, \ xy,\  \frac 12 y^2,\  -\frac 12 x^2, \ 1}$&$\dd x\wedge \dd y$  &$\overline{\mathfrak{sl}(2 )\ltimes \mathbb{R}^2}\simeq \mathfrak{h}_6$    \\[+1.5ex]
\hline
 &&&  &\\[-1.5ex]
\# & Imprimitive &  Hamiltonian functions $h_i$& $\omega$  &  Lie--Hamilton algebra \\[+1.0ex]
\hline
 &  &\\[-1.5ex]
I$_1$& $\mathbb{R}$ &$\int^y{f(y')\dd y'}$ & $f(y)\dd x\wedge \dd y$ & $ \mathbb{R}$ or $\mathbb{R}^2$\\[+1.0ex]
I$_4$& $\mathfrak{sl}(2 )$ (type II) & $\displaystyle{ \frac 1 {{x-y}   } ,\ \frac{x+y}{2(x-y)},\ \frac{xy}{x-y} }$ &$\displaystyle   {\frac{\dd x\wedge \dd y} {{(x-y)^{2}}}}$ &$\mathfrak{sl}(2 )$ or $\mathfrak{sl}(2)\oplus\mathbb{R}$\\[+2.0ex]
I$_5$& $\mathfrak{sl}(2 )$ (type III) &$\displaystyle { {-\frac{1}{2y^{2}},\ -\frac{x}{y^{2}},\ -\frac{x^2}{2 y^{2}  }  }} $  &$\displaystyle{\frac{\dd x\wedge \dd y}{y^3}}$ &$\mathfrak{sl}(2 )$ or  ${\mathfrak{sl}}(2)\oplus\mathbb{R}$\\[+2.0ex]
I$_8$& $B_{-1}\simeq  { {\mathfrak{iso}}}(1,1)$ &${y,\ -x,\ xy,\ 1 }$ & $\dd x\wedge \dd y$ & ${ \overline{\mathfrak{iso}}}(1,1) \simeq \mathfrak{h}_4$\\[+2.0ex]
I$_{12}$& $\mathbb{R}^{r+1}$ &$-\int^x \!\! {f(x')\dd x'}, - \int^x \!\! {f(x')\xi_j(x')\dd x'}$ & $f(x)\dd x\wedge \dd y$ & $\mathbb{R}^{r+1}$ or $\mathbb{R}^{r+2}$ \\[+2.0ex]
I$_{14A}$& $\mathbb{R} \ltimes \mathbb{R}^{r}$ (type I) &$y,\  - \int^x {\eta_j(x')\dd x'} $,\quad $1\notin  \langle \eta_j \rangle$ & $ \dd x\wedge \dd y$ & $\mathbb{R}\ltimes\mathbb{R}^{r}$ or $(\mathbb{R}\ltimes \mathbb{R}^{r})\oplus\mathbb{R}$\\[+2.0ex]
I$_{14B}$& $\mathbb{R} \ltimes \mathbb{R}^{r}$ (type II) &$y,\ -x, \ - \int^x{\eta_j(x')\dd x'},\ 1 $ & $ \dd x\wedge \dd y$ & $ \overline{(\mathbb{R} \ltimes \mathbb{R}^{r}) }$\\[+2.0ex]
I$_{16}$& $C_{-1}^r \simeq  {\mathfrak{h}_2\ltimes\mathbb{R}^{r+1}}$ & $\displaystyle{ {y,\  -x, \ xy,  \  -\frac{x^{j+1}}{j+1} ,\ 1} }$ & $\dd x\wedge \dd y$ & $ \overline{\mathfrak{h}_2\ltimes\mathbb{R}^{r+1}} $\\[+2.0ex]
    \hline
 \end{tabular}
\hfill}
\end{table}

\subsection{Primitive  Lie--Hamilton algebras}

\subsubsection{Lie algebra P$^{(\alpha=0)}_{1}$: $A_0\simeq  {\mathfrak{iso}}(2) $}

Proposition \ref{GoSym} states that $A_0$ is a Lie algebra of Hamiltonian vector fields with respect to the symplectic form $\omega=\dd x\wedge \dd y$. The basis of vector fields $X_1, X_2, X_3$ of  $A_0$ given in Table \ref{table1} satisfy the commutation relations
 $$
 [X_1,X_2]=0,\qquad [X_1,X_3]= -X_2,\qquad [X_2,X_3]=X_1.
 $$
So, $A_0$ is isomorphic to the  two-dimensional Euclidean algebra $ {\mathfrak{iso}}(2) $.
Using the relation $\iota_{X}\omega={\rm d}h$ between a Hamiltonian vector field and one of its Hamiltonian functions, we get that the Hamiltonian functions for $X_1,X_2,X_3$ read
$$
h_1=y,\qquad h_2=-x,\qquad h_3=\tfrac 12(x^2+y^2),
$$
correspondingly. Along with $h_0=1$, these functions  fulfil 
$$
\{h_1,h_2\}_\omega=h_0,\quad \{h_1,h_3\}_\omega=h_2,\quad \{h_2,h_3\}_\omega=-h_1,\quad  \{h_0,\cdot\}_\omega=0 .
$$
Consequently, if $X$ is a Lie--Hamilton system admitting a minimal Lie algebra $A_0$, i.e. $X=\sum_{i=1}^3b_i(t)X_i$ for certain $t$-dependent functions $b_1,b_2,b_3$ such that $V^X\simeq A_0$, then it admits a Lie--Hamiltonian structure $(U,\omega,h=\sum_{i=1}^3b_i(t)h_i$) and a Lie--Hamilton algebra $(\ham,\{\cdot,\cdot\}_\omega)$ generated by the functions $\langle h_1,h_2,h_3,h_0\rangle$. Hence, $(\ham,\{\cdot,\cdot\}_\omega)$ is a finite-dimensional real Lie algebra of Hamiltonian functions isomorphic to the {\em centrally extended}  Euclidean algebra $\overline{\mathfrak{iso}}(2) $~\cite{azca}.
Note that $1\in \{\mathcal{H}_\Lambda,\mathcal{H}_\Lambda\}_\omega$. In virtue of Corollary \ref{NoGo2}, system $X$ does not admit any Lie--Hamilton algebra isomorphic to $V^X$. Moreover, Proposition~\ref{UniExt} ensures that all Lie--Hamilton algebras for $X$ are isomorphic to $\overline{\mathfrak{iso}}(2)$.

\subsubsection{Lie algebra P$_2$: $\mathfrak{sl}(2 )$}

We have already proved in Section \ref{BLHS} that the Lie algebra of vector fields P$_2$, which is spanned by the vector fields (\ref{vectRiccati2}), is a Lie algebra of Hamiltonian vector fields with respect to the symplectic structure (\ref{aa}).
The   Hamiltonian functions $h_1,h_2,h_3$ for $X_1$, $X_2$ and $X_3$  were found to be (\ref{ab}), correspondingly. Then, a Lie system $X$ with minimal Lie algebra P$_2$, i.e. a system of the form $X=\sum_{i=1}^3b_i(t)X_i$ for certain $t$-dependent functions $b_1,b_2,b_3$ such that $V^X={\rm P}_2$, is a Lie--Hamilton system
with respect to the Poisson bracket induced by (\ref{aa}). Then, $X$ admits a Lie--Hamiltonian structure $(U,\omega,h=\sum_{i=1}^3b_i(t)h_i)$ and  a Lie--Hamilton algebra isomorphic to $\mathfrak{sl}(2)$ with commutation relations (\ref{sl2Rh}). In view of Proposition \ref{LieHam2}, any Lie--Hamilton system associated to P$_2$ also admits a Lie--Hamilton algebra isomorphic to $\mathfrak{sl}(2)\oplus\mathbb{R}$. In view of Proposition \ref{LieHam1}, these are the only algebraic structures of the Lie--Hamilton algebras for such Lie--Hamilton systems.

\subsubsection{Lie algebra P$_3$: $\mathfrak{so}(3)$}
In this case, we must determine a symplectic structure $\omega$ turning the elements of the modular generating system $(U_1,X_1,X_2)$ of P$_3$ into locally Hamiltonian vector fields with respect to a symplectic structure $\omega$ (\ref{ww}). In view of Theorem \ref{Char}, this ensures that every element of P$_3$ is Hamiltonian with respect to $\omega$. The condition $\mathcal{L}_{X_1}\omega=0$ gives
$$
 y\frac{\partial f}{\partial x}-x\frac{\partial f}{\partial y} =0.
$$
Applying the characteristics method, we find that $f$ must be constant along the integral curves of the system
$x\, \dd x+ y\, \dd y =0$,  namely curves with $x^2+y^2=k\in\mathbb{R}$. So, $f=f(x^2+y^2)$. If we now require  $\mathcal{L}_{X_2}\omega=0$, we obtain that
$$
 (1+x^2-y^2)\frac{\partial f}{\partial x}+2xy\frac{\partial f}{\partial y}+4xf=0.
$$
Using that $f=f(x^2+y^2)$, we have
$$
\frac{f'}{f}=-\frac{2}{1+x^2+y^2} \Rightarrow f(x^2+y^2)=(1+x^2+y^2)^{-2}.
$$
  Then,
  $$
\omega=\frac{\dd x\wedge \dd y}{(1+x^2+y^2)^2}.
$$
So, P$_3$ becomes a Lie algebra of Hamiltonian vector fields relative to $\omega$. The vector fields $X_1$, $X_2$ and $X_3$ admit the Hamiltonian functions
$$
h_1=-\frac{1}{2(1+x^2+y^2)},\qquad h_2=\frac{y}{1+x^2+y^2},\qquad h_3=-\frac{x}{1+x^2+y^2},
$$
which along $h_0=1$ satisfy the commutation relations
$$
\{h_1,h_2\}_\omega=-h_3,\qquad \{h_1,h_3\}_\omega= h_2 ,\qquad \{h_2,h_3\}_\omega=-4h_1-h_0,\qquad   \{h_0,\cdot\}_\omega=0,
$$
with respect to the Poisson bracket induced by $\omega$. Then, $\langle h_1,h_2,h_3, h_0\rangle$ span a Lie algebra of Hamiltonian functions isomorphic to a central extension of  ${\mathfrak{so}}(3)$, denoted $\overline{\mathfrak{so}}(3)$. It is well known  \cite{azca}   that the central extension associated with $h_0$ is a trivial one; if we define $\bar h_1= h_1+h_0/4$, then $\langle \bar h_1,h_2,h_3\rangle$ span a Lie algebra  isomorphic to ${\mathfrak{so}}(3)$ and $\overline{\mathfrak{so}}(3)\simeq \mathfrak{so}(3)\oplus\mathbb{R}$. In view of this and using Propositions \ref{LieHam1} and \ref{LieHam2}, a Lie system admitting a minimal Lie algebra P$_3$ admits Lie--Hamilton structures isomorphic to $\mathfrak{so}(3)\oplus\mathbb{R}$ and $\mathfrak{so}(3)$.

\subsubsection{Lie algebra P$_5$: $\mathfrak{sl}(2 )\ltimes \mathbb{R}^2$}

From Proposition \ref{GoSym}, this Lie algebra consists of Hamiltonian vector fields with respect to the symplectic form $\omega=\dd x\wedge \dd y$. The vector fields of the basis
$X_1,\ldots,X_5$ for P$_5$ given in Table~\ref{table1} are Hamiltonian vector fields relative to $\omega$ with  Hamiltonian functions
$$
h_1=y,\qquad h_2=-x,\qquad h_3=xy,\qquad h_4=\tfrac 12 {y^2},\qquad h_5=-\tfrac
 12 x^2,
$$
correspondingly. These functions together with $h_0=1$ satisfy the relations
$$
\begin{array}{llll}
\{h_1,h_2\}_\omega=h_0, &\quad \{h_1,h_3\}_\omega=-h_1,&\quad \{h_1,h_4\}_\omega=0,&\quad \{h_1,h_5\}_\omega=-h_2,\\[2pt]
\{h_2,h_3\}_\omega=h_2,&\quad\{h_2,h_4\}_\omega=-h_1,&\quad \{h_2,h_5\}_\omega=0,&\quad
\{h_3,h_4\}_\omega=2h_4, \\[2pt]
\{h_3,h_5\}_\omega=-2h_5,&\quad
\{h_4,h_5\}_\omega=h_3, &\quad   \{h_0,\cdot\}_\omega=0 . &
\end{array}
$$
Hence $\langle h_1,\ldots,h_5  , h_0\rangle$ span a Lie algebra $\overline{\mathfrak{sl}(2 )\ltimes \mathbb{R}^2}$ which is a non-trivial central extension of P$_5$, i.e.~it is not isomorphic to P$_5\oplus\mathbb{R}$. In fact,   it is isomorphic to the so called two-photon Lie algebra   $\mathfrak{h}_6$ (see~\cite{BBF09} and references therein); this  can be brought into the form $\mathfrak{h}_6\simeq \mathfrak{sl}(2)
\oplus_s \mathfrak{h}_3$,
where $\mathfrak{sl}(2 ) \simeq \langle h_3,h_4,h_5\rangle $,
$\mathfrak{h}_3\simeq  \langle h_1,h_2,h_0\rangle$ is the Heisenberg--Weyl Lie algebra, and
$\oplus_s$ stands for a semidirect sum. Furthermore, $\mathfrak{h}_6$ is also  isomorphic to the  $(1+1)$-dimensional centrally extended Schr\"odinger  Lie algebra~\cite{Schrod}.

In view of Corollary \ref{NoGo2}, Proposition \ref{UniExt} and following the same line of reasoning than in previous cases, a Lie system admitting a minimal Lie algebra P$_5$ only possesses Lie--Hamilton algebras isomorphic 	to $\mathfrak{h}_6$.


\subsection{Imprimitive  Lie--Hamilton algebras}

\subsubsection{Lie algebra I$_1$: $\mathbb{R}$}
Note that $X_1=\partial_x$ is a modular generating basis of I$_1$. By solving the PDE $\mathcal{L}_{X_1}\omega=0$ with $\omega$ written in the form (\ref{ww}), we obtain that $\omega=f(y)\dd x\wedge \dd y$ with $f(y)$ being any non-vanishing function of $y$. In view of Theorem \ref{Char}, the Lie algebra I$_1$ becomes a Lie algebra of Hamiltonian vector fields with respect to $\omega$. Observe that $X_1$, a basis of I$_1$, has a Hamiltonian function, $h_1=\int^y f(y')\dd y'$. As $h_1$ spans a Lie algebra isomorphic to $\mathbb{R}$, it is obvious that a system $X$ with $V^X\simeq I_1$ admits a Lie--Hamilton algebra isomorphic to I$_1$. Proposition \ref{LieHam2} yields that $X$ admits a Lie--Hamilton algebra isomorphic to $\mathbb{R}^2$. In view of Proposition \ref{LieHam1}, these are the only algebraic structures for the Lie--Hamilton algebras for $X$.

\subsubsection{Lie algebra I$_4$: $\mathfrak{sl}(2)$ of type II}
This Lie algebra admits a modular generating system $(\mathbb{R}^2_{x\neq y},X_1= {\partial}_x+  {\partial}_y, X_2= x{\partial}_x +y{\partial}_y)$.
Let us search for a symplectic form $\omega$ (\ref{ww})  turning $X_1$ and $X_2$ into local Hamiltonian vector fields with respect to it. Using Theorem \ref{Char}, we can ensure that every element of I$_4$ is Hamiltonian with respect to $\omega$.  By imposing $\mathcal{L}_{X_i}\omega=0$ ($i=1,2$), we find that
$$
\frac{\partial f}{\partial x}+\frac{\partial f}{\partial y}=0,\qquad x\frac{\partial f}{\partial x}+y\frac{\partial f}{\partial y}+2 f=0.
$$
 Applying the method of characteristics to the first equation, we have that
 $\dd x=\dd y $.  Then $f=f(x-y)$. Using this in the second equation, we obtain  a particular solution $f= (x-y)^{-2}$ which  gives rise to a  closed and non-degenerate two-form, namely
\begin{equation}
\omega=\frac{{\rm d} x \wedge {\rm d} y}{(x -y)^2}  .
\label{aa2}
\end{equation}
Hence,
 \begin{equation*}
 h_1=\frac{1}{x-y} ,\qquad
h_2=  \frac{x +y}{2(x -y)}  ,\qquad
h_3=\frac{xy}{x -y}
\end{equation*}
are Hamiltonian functions of the vector fields $X_1,X_2,X_3$ of the basis for I$_4$ given in Table \ref{table1}, respectively. Using the Poisson bracket $\{\cdot,\cdot\}_\omega$ induced by (\ref{aa2}), we obtain that $h_1,h_2$ and $h_3$ satisfy
\begin{equation*}
\{h_1,h_2\}_\omega=-h_1,\qquad \{h_1,h_3\}_\omega=-2h_2,\qquad \{h_2,h_3\}_\omega=-h_3.
\end{equation*}
Then, $(\langle h_1,h_2,h_3\rangle,\{\cdot,\cdot\}_\omega) \simeq \mathfrak{sl}(2)$.
Consequently, if $X$ is a Lie--Hamilton system admitting a minimal Lie algebra I$_4$, it admits a Lie--Hamilton algebra that is isomorphic to  $\mathfrak{sl}(2)$ or, from Proposition \ref{LieHam2}, to $\mathfrak{sl}(2)\oplus\mathbb{R}$. From Proposition \ref{LieHam1}, these are the only algebraic structures for its Lie--Hamilton algebras.

\subsubsection{Lie algebra I$_5$: $\mathfrak{sl}(2)$ of type III}
Observe that $(U,X_1=\partial_x,X_2=2x\partial_x+y\partial_y)$ form a modular generating system of I$_5$. The conditions $\mathcal{L}_{X_1}\omega=\mathcal{L}_{X_2}\omega=0$ ensuring that $X_1$ and $X_2$ are locally Hamiltonian with respect to $\omega$ give rise to the equations
 $$
 \frac{\partial f}{\partial x}=0,\qquad 2x\frac{\partial f}{\partial x}+y\frac{\partial f}{\partial y}+3f=0,
 $$
so that  $f(x,y)=y^{-3}$ and $X_1,X_2$ become locally Hamiltonian vector fields relative to the symplectic form
 $$
\omega=\frac{\dd x\wedge \dd y}{y^3} .
$$
In view of Theorem \ref{Char}, this implies that every element of I$_5$ is Hamiltonian with respect to $\omega$.
  Hamiltonian functions for the elements of the basis $X_1,X_2,X_3$ for I$_5$ given in Table \ref{table1} read
$$
h_1=-\frac{1}{2y^{2}},\qquad  h_2= -\frac{x}{y^{2}},\qquad  h_3=-\frac{x^2}{2 y^{2}  }.
$$
They span a Lie algebra isomorphic to $\mathfrak{sl}(2)$:
$$
\{h_1,h_2\}_\omega=-2h_1,\qquad \{h_1,h_3\}_\omega=-h_2,\qquad \{h_2,h_3\}_\omega=-2h_3.
$$
Therefore, a Lie system possessing a minimal Lie algebra I$_5$ possesses a Lie--Hamilton algebra isomorphic to $\mathfrak{sl}(2)$ and, in view of Proposition \ref{LieHam2}, to $\mathfrak{sl}(2)\oplus\mathbb{R}$. In view of Proposition \ref{LieHam1}, these are the only possible algebraic structures for the Lie--Hamilton algebras for $X$.

\subsubsection{Lie algebra I$_8^{(\alpha=-1)}$: $B_{-1}\simeq{\mathfrak{iso}}(1,1) $}

In view of Proposition \ref{GoSym}, this Lie algebra consists of Hamiltonian vector fields with respect to the standard symplectic structure $\omega=\dd x\wedge \dd y$.
The elements of the basis for $B_{-1}$ detailed in Table \ref{table1} satisfy the commutation relations
 $$
 [X_1,X_2]=0,\qquad [X_1,X_3]= X_1,\qquad [X_2,X_3]=-X_2 .
 $$
Hence, these vector fields
 span a Lie  algebra isomorphic to the  (1+1)-dimensional Poincar\'e algebra $ {\mathfrak{iso}}(1,1) $. Their   corresponding  Hamiltonian functions turn out to be
$$
h_1=y,\qquad h_2=-x,\qquad h_3=xy ,
$$
which together with a central generator $h_0=1$ fulfil the commutation relations
$$
\{h_1,h_2\}_\omega=h_0,\qquad \{h_1,h_3\}_\omega=-h_1,\qquad \{h_2,h_3\}_\omega=h_2,\qquad \{h_0,\cdot\}_\omega=0 .
$$
Thus, a Lie  system $X$ admitting a minimal Lie algebra $B_{-1}$ possesses a Lie--Hamilton algebra isomorphic to the centrally extended Poincar\'e algebra $ {\overline{\mathfrak{iso}}}(1,1) $ which, in turn, is also isomorphic to the harmonic oscillator    algebra $\mathfrak{h}_4$. As is well  known~\cite{azca}, this Lie algebra is not of the form $\mathfrak{iso}(1,1)\oplus\mathbb{R}$, then Proposition \ref{LieHam1} ensures that $X$ does not admit any Lie--Hamilton algebra isomorphic to $\mathfrak{iso}(1,1)$. Moreover, Proposition \ref{UniExt} states that all Lie--Hamilton algebras of $X$ must be isomorphic to $ {\overline{\mathfrak{iso}}}(1,1)$.

\subsubsection{Lie algebra I$_{12}$: $\mathbb{R}^{r+1}$}

The vector field $X_1=\partial_y$ is a modular generating system for I$_{12}$ and all the elements of this Lie algebra have zero divergence. By solving the PDE $\mathcal{L}_{X_1}\omega=0$, where we recall that $\omega$ has the form (\ref{ww}),  we see that $f=f(x)$ and $X_1$ becomes Hamiltonian for any non-vanishing function $f(x)$. In view of Theorem \ref{Char}, the remaining elements of I$_{12}$ become automatically Hamiltonian with respect to $\omega$. Then, we obtain that $X_1,
\ldots,X_{r+1}$ are Hamiltonian vector fields relative to the symplectic structure $\omega=f(x)\dd x\wedge \dd y$ with   Hamiltonian functions
$$
h_1=- \int^{x}{f(x')\dd x'} ,\qquad  h_{j+1} =-  \int^{x} f(x')\xi_j(x')\dd x' ,\qquad j=1,\ldots,r ,\qquad r\ge 1,
$$
which span the Abelian Lie algebra $\mathbb{R}^{r+1}$. In consequence, a Lie--Hamilton system $X$ related to a minimal Lie algebra I$_{12}$ possesses a Lie--Hamilton algebra isomorphic to $\mathbb{R}^{r+1}$. From Propositions \ref{LieHam1} and \ref{LieHam2}, it only admits an additional Lie--Hamilton algebra isomorphic to $\mathbb{R}^{r+2}$.

 \subsubsection{I$_{14}$: $\mathbb{R}\ltimes\mathbb{R}^{r}$}
The functions $\eta_1(x),\ldots,\eta_r(x)$ form a fundamental system of solutions of a system of $r$ first-order differential equations with constant coefficients \cite{GKP92,HA75}. Hence, none of them vanishes in an open interval of $\mathbb{R}$ and I$_{14}$ is such that
$(U_1,X_1,X_2)$, where $X_1$ and $X_2$ are given in Table \ref{table1},  form a modular generating system. Since all the elements of I$_{14}$ have divergence zero and using Theorem \ref{Char}, we infer that I$_{14}$  consists of Hamiltonian vector fields relative to a symplectic structure if and only if $X_1$ and $X_2$ are locally Hamiltonian vector fields with respect to a symplectic structure. By requiring $\mathcal{L}_{X_i}\omega=0$, with $i=1,2$ and $\omega$ of the form (\ref{ww}), we obtain that
$$
 \frac{\partial f}{\partial x}=0,\qquad \eta_j(x)  \frac{\partial f}{\partial y}=0,\qquad j=1,\ldots,r.
$$
So, I$_{14}$ is only compatible with $\omega=\dd x\wedge \dd y$. The Hamiltonian functions corresponding to $X_1,\ldots,X_{r+1}$ turn out to be
\begin{equation}\label{hh}
h_1=y ,\qquad  h_{j+1} =-  \int^{x} \eta_j(x')\dd x' ,\qquad j=1,\ldots,r ,\qquad r\ge 1.
\end{equation}
We remark that two {\em different} Lie--Hamiltonian algebras, denoted I$_{14A}$ and I$_{14B}$,  are spanned by the above Hamiltonian functions:

\begin{itemize}

\item I$_{14A}$: If $1\notin \langle \eta_1,\ldots,\eta_{r}\rangle$,  then the functions (\ref{hh}) span a Lie algebra $\mathbb{R}\ltimes\mathbb{R}^{r}$ and, by considering Propositions~\ref{LieHam1} and \ref{LieHam2}, this case only admits an additional Lie--Hamilton algebra isomorphic to  $(\mathbb{R}\ltimes \mathbb{R}^r)\oplus\mathbb{R}$.

\item I$_{14B}$: If $1\in \langle \eta_1,\ldots,\eta_{r}\rangle$, we can choose a basis of I$_{14}$ in such a way that there exists a function, let us say $\eta_1$, equal to $1$. Then the Hamiltonian functions (\ref{hh}) turn out to be
$$
h_1=y ,\qquad  h_2=-x,\qquad h_{j+1} =-  \int^{x} \eta_j(x')\dd x' ,\qquad j=2,\ldots,r ,\qquad r\ge 1,
$$
 which require a central generator $h_0=1$  in order to close   a centrally extended Lie algebra $(\ham,\{\cdot,\cdot\}_\omega)\simeq  \overline{(\mathbb{R} \ltimes \mathbb{R}^{r})}$.

\end{itemize}

In view of the above, a Lie system $X$ with a minimal Lie algebra I$_{14}$ is a Lie--Hamiltonian system. Its Lie--Hamilton algebras can be isomorphic to I$_{14}$ or I$_{14}\oplus \mathbb{R}$ when  $1\notin \langle \eta_1,\ldots,\eta_{r}\rangle$ (class I$_{14A}$). If $1\in \langle \eta_1,\ldots,\eta_{r}\rangle$ (class I$_{14B}$), a Lie--Hamilton algebra is isomorphic to $\overline{\mathbb{R} \ltimes \mathbb{R}^{r}}$ and since  $1\in \{\ham ,\ham\}_\omega$, we obtain from Corollary \ref{NoGo2} and Proposition \ref{UniExt} that every Lie--Hamilton algebra for $X$ is isomorphic to it.

\subsubsection{Lie algebra I$^{(\alpha=-1)}_{16}$: $C_{-1}^r \simeq \mathfrak{h}_2\ltimes\mathbb{R}^{r+1}$}

In view of Proposition \ref{GoSym}, this Lie algebra consists of Hamiltonian vector fields with respect to the standard symplectic structure. The resulting  Hamiltonian functions for $X_1,\ldots, X_{r+3}$ are given by
   $$
h_1=y,\qquad h_2=-x,\qquad h_3=xy ,\qquad h_{j+3}=-\frac{x^{j+1}}{j+1} ,\qquad j=1,\dots,r, \qquad r\ge 1,
$$
which again require an additional central generator $h_0=1$ to close on a finite-dimenisonal Lie algebra. The commutation relations for this Lie algebra are given by
$$
\begin{array}{llll}
&\{h_1,h_2\}_\omega=h_0,\quad &\{h_1,h_3\}_\omega=-h_1,\quad &\{h_2,h_3\}_\omega=h_2,\\[2pt]
&\{h_1,h_{4}\}_\omega=-h_{2},\quad  &\{h_1,h_{k+4}\}_\omega=-( k +1)h_{k+3},\quad &\{h_2,h_{j+3}\}_\omega=0  ,\\[2pt]
&  \{h_3,h_{j+3} \}_\omega=-(j+1) h_{j+3},\quad   &\{h_{j+3} ,h_{k+4}\}_\omega= 0,  \quad &\{h_0,\cdot\}_\omega=0 ,
\end{array}
$$
with $j=1,\dots, r$ and $k=1,\dots,r-1$, which define the centrally extended Lie algebra $ \overline{\mathfrak{h}_2\ltimes\mathbb{R}^{r+1}}$. 

Then, given a Lie system $X$ with a minimal Lie algebra $C_{-1}^r$, the system is a Lie--Hamiltonian one which admits a Lie--Hamilton algebra isomorphic to $\overline{\mathfrak{h}_2\ltimes \mathbb{R}^{r+1}}$. As $1\in \left\{\overline{\mathfrak{h}_2\ltimes\mathbb{R}^{r+1}},\overline{\mathfrak{h}_2\ltimes\mathbb{R}^{r+1}}\right\}_\omega$, Corollary \ref{NoGo2} and Proposition \ref{UniExt} ensure that  every Lie--Hamilton algebra for $X$ is isomorphic to $\overline{\mathfrak{h}_2\ltimes\mathbb{R}^{r+1}}$.


\section{Application to $\mathfrak{sl}(2)$-Lie systems on the plane}

In this  section, we employ our techniques to study the properties of certain Lie systems on the plane admitting a Vessiot--Guldberg Lie algebra isomorphic to $\mathfrak{sl}(2)$, the so-called {\it $\mathfrak{sl}(2)$-Lie systems} \cite{LS12,Pi12}. More specifically, we analyse in detail  Lie systems   used to describe Milne--Pinney equations \cite{SIGMA}, Kummer--Schwarz equations \cite{CGL11} and complex Riccati equations with real $t$-dependent coefficients  \cite{Eg07}. As a byproduct, our results also cover the 
$t$-dependent frequency harmonic oscillator.


\subsection{Milne--Pinney equations}
The Milne--Pinney equation, which is well known for its multiple properties and applications in physics (see~\cite{LA08} and references therein), takes the form
$$
\frac{\dd^2x}{\dd t^2}=-\omega^2(t)x+\frac{c}{x^3},
$$
where $\omega(t)$ is any $t$-dependent function and $c$ is a real constant.
By adding a new variable $y\equiv \dd x/\dd t$, we can study these equations through the first-order system
\begin{equation}\label{FirstLie}
\left\{
\begin{aligned}
\frac{\dd x}{\dd t}&=y,\\
\frac{\dd y}{\dd t}&=-\omega^2(t)x+\frac{c}{x^3},
\end{aligned}\right.
\end{equation}
which is a Lie system \cite{PW, SIGMA}. We recall that (\ref{FirstLie}) can be regarded as  the equations of motion of the one-dimensional Smorodinsky--Winternitz system~\cite{BCHLS13Ham, WSUF65}; moreover, when the parameter $c$ vanishes, this reduces to the harmonic oscillator (both with   a $t$-dependent frequency).
Explicitly, (\ref{FirstLie}) is the associated system to the $t$-dependent vector field
$
X_t=X_3+\omega^2(t)X_1,
$
where
\begin{equation}\label{FirstLieA}
X_1=-x\frac{\partial}{\partial y},\qquad X_2=\frac 12 \left(y\frac{\partial}{\partial y}-x\frac{\partial}{\partial x}\right),\qquad X_3=y\frac{\partial}{\partial x}+\frac{c}{x^3}\frac{\partial}{\partial y},
\end{equation}
span a finite-dimensional real Lie algebra $V$ of vector fields isomorphic to $\mathfrak{sl}(2)$ with commutation relations given by
\begin{equation}\label{FirstLieB}
[X_1,X_2]=X_1,\qquad [X_1,X_3]=2 X_2,\qquad  [X_2,X_3]=X_3 .
\end{equation}

There are {\em  four} classes of finite-dimensional Lie algebras of vector fields isomorphic to $\mathfrak{sl}(2)$ in the GKO  classification: P$_2$ and I$_3$--I$_5$. To determine which one is locally diffeomorphic to $V$, we first  find out whether $V$ is imprimitive or not. In this respect,
 recall that $V$ is  {\em imprimitive} if there exists a one-dimensional distribution $\mathcal{D}$ invariant under the action (by Lie brackets) of the elements of $V$. Hence,  $\mathcal{D}$  is spanned by
 a non-vanishing vector field
$$
Y=\fff_x(x,y)\frac{\partial}{\partial x}+\fff_y(x,y)\frac{\partial}{\partial y},
$$
 which must be invariant under the action of $X_1$, $X_2$ and $X_3$. As $\fff_x$ and $\fff_y$ cannot vanish simultaneously, $Y$ can be taken either of the following local forms
\begin{equation}\label{xxx}
Y=\frac{\partial}{\partial x}+\fff_y\frac{\partial}{\partial y},\qquad Y=\fff_x\frac{\partial}{\partial x}+\frac{\partial}{\partial y}.
\end{equation}
Let us assume that $\mathcal{D}$ is spanned by the first one and search for   $Y$. Now, if $\mathcal{D}$ is invariant under the Lie brackets of the elements of $V$, we have that
\begin{subequations}
\begin{align}
&\quad\qquad \qquad
 \mathcal{L}_{X_1}Y=\left(1-x\frac{\partial \fff_y}{\partial y}\right)\frac{\partial}{\partial y}=\hh_1 Y,\label{con2}\\
&\mathcal{L}_{X_2}Y=\frac 12\left[\frac{\partial}{\partial x}+\left(y\frac{\partial \fff_y}{\partial y}-x\frac{\partial \fff_y}{\partial x}- \fff_y\right)\frac{\partial}{\partial y}\right]=\hh_2 Y,\label{con1}\\
&\mathcal{L}_{X_3}Y=-\fff_y\frac{\partial}{\partial x}+\left(\frac{3c}{x^4}+y\frac{\partial \fff_y}{\partial x}+\frac c {x^3}\frac{\partial \fff_y}{\partial y}\right)\frac{\partial}{\partial y}=\hh_3 Y,\label{con3}
\end{align}
\end{subequations}
for certain functions $\hh_1,\hh_2,\hh_3$ locally defined on $\mathbb{R}^2$. The left-hand side of (\ref{con2}) has no term $\partial_x$ but the right-hand one has it provided $\hh_1\neq 0$. Therefore, $\hh_1=0$ and $\fff_y= { y}/{x}+G$ for a certain $G=G(x)$. Next by introducing this result in
  (\ref{con1}), we find that $\hh_2=1/2$ and $2G + x    G^\prime =0$, so that $G(x)=\mm/x^2$ for   $\mm\in\mathbb{R}$.
   Substituting this into (\ref{con3}), we obtain that $\hh_3=-(\mm+x y)/x^2$
and $ \mm^2=-4c$. Consequently, when  $c>0$  it does not exist any non-zero $Y$ spanning locally $\mathcal{D}$ satisfying (\ref{con2})--(\ref{con3})  and $V$ is therefore primitive, whilst   if $c\le 0$   there exists a vector field
$$
Y= \frac{\partial}{\partial x}+\left(\frac{y}{x} +\frac{\mm}{x^2}\right) \frac{\partial}{\partial y},\qquad \mm^2=-4c,
$$
   which spans  $\mathcal{D}$, so that $V$ is  imprimitive.   The case of $\mathcal{D}$ being   spanned by the second form of $Y$   (\ref{xxx}) can be analysed analogously and drives to the same conclusion. 

Therefore the system     (\ref{FirstLie})  belongs to different classes 
within the GKO classification according to  the value of the parameter $c$. The final result   is established in the following statement.

\begin{proposition}\label{prop71}  The system   (\ref{FirstLie}), corresponding to the the Milne--Pinney equations,  is locally diffeomorphic to   {\rm P}$_2$ for $c>0$, {\rm I}$_4$ for $c<0$ and {\rm I}$_5$ for $c=0$.
\end{proposition}
\begin{proof}  Since $V$ is  primitive when $c>0$ and  this is  isomorphic
to $\mathfrak{sl}(2)$,  the GKO classification given in Table 1 implies  that $V$ is locally diffeomorphic to
the primitive class P$_2$.

 Let us now consider that $c<0$ and prove that the system is
 diffeomorphic to  the class  I$_4$. We do this by showing that there exists  a local diffeomorphism 
$\phi:(x,y)\in U\subset \mathbb{R}^2_{x\ne y}\mapsto \bar U\subset (u,v)\in\mathbb R^2_{u\neq 0}$, satisfying that $\phi_*$ maps the basis  for I$_4$  listed in Table 1 into
 (\ref{FirstLieA}). Due to the Lie bracket $[X_1,X_3]=2X_2$, verified in both bases,  it is only necessary to search    the map for the generators $X_1$ and $X_3$  (so for  $X_2$ this will be automatically fulfilled). 
  By writing  in coordinates
$$
\phi_*(\partial_x+\partial_y)=-x\partial_y,\qquad  \phi_*(x^2\partial_x+y^2\partial_y)=y\partial_x+c/x^3\partial_y ,
$$
we obtain
$$
\left(
\begin{array}{cc}
\frac{\partial u}{\partial x}&\frac{\partial u}{\partial y}\\[4pt]
\frac{\partial v}{\partial x}&\frac{\partial v}{\partial y}\\
\end{array}
\right)\left(\begin{array}{c}
              1\\[4pt]
              1
             \end{array}\right)=
             \left(
             \begin{array}{c}
              0\\[4pt] -u
             \end{array}\right),\qquad
             \left(
\begin{array}{cc}
\frac{\partial u}{\partial x}&\frac{\partial u}{\partial y}\\[4pt]
\frac{\partial v}{\partial x}&\frac{\partial v}{\partial y}\\
\end{array}
\right)\left(\begin{array}{c}
              x^2\\[4pt]
              y^2
             \end{array}\right)=
             \left(
             \begin{array}{c}
              v\\ [4pt]  {c}/{u^3}
             \end{array}\right).
$$
Hence, ${\partial_x u}+{\partial_y u}=0\Rightarrow u=\ff(x-y)$ for a certain $\ff:z\in\mathbb{R}\mapsto g(z)\in \mathbb{R}$. Since $x^2\partial_x u +y^2\partial_y u  =v$, then $v=(x^2-y^2) \ff'$, where $\ff'$ is the derivative of $\ff(z)$ in terms of $z$.
Using now that ${\partial_x v}+{\partial_y v}=-u$ we get $2(x-y)\ff^\prime=-\ff$ so that $\ff=\lambda/ |x-y|^{1/2}$  where $\lambda\in\mathbb{R}\backslash \{0\}$. Substituting  this into the remaining equation $x^2\partial_x v +y^2\partial_y v  =c/u^3$,  we find that $\lambda^4=-4 c$. 
Since $c<0$, we consistently find that 
$$
u=\frac{\lambda}{|x-y|^{1/2}},\qquad v=-\frac{\lambda(x+y )}{2|x-y|^{1/2}},\qquad \lambda^4=-4 c.
$$
Finally,  let us set $c=0$ and search for a   local diffeomorphism 
$\phi:(x,y)\in U\subset \mathbb{R}^2_{y\ne 0}\mapsto \bar U\subset (u,v)\in\mathbb R^2$ such that $\phi_*$ maps the basis corresponding to I$_5$ into (\ref{FirstLieA}); namely
$$
\phi_*(\partial_x )=-x\partial_y,\qquad  \phi_*(x^2\partial_x+xy\partial_y)=y\partial_x,
$$
yielding
$$
\left(
\begin{array}{cc}
\frac{\partial u}{\partial x}&\frac{\partial u}{\partial y}\\[4pt]
\frac{\partial v}{\partial x}&\frac{\partial v}{\partial y}\\
\end{array}
\right)\left(\begin{array}{c}
              1\\[4pt]
              0
             \end{array}\right)=
             \left(
             \begin{array}{c}
              0\\[4pt] -u
             \end{array}\right),\qquad
             \left(
\begin{array}{cc}
\frac{\partial u}{\partial x}&\frac{\partial u}{\partial y}\\[4pt]
\frac{\partial v}{\partial x}&\frac{\partial v}{\partial y}\\
\end{array}
\right)\left(\begin{array}{c}
              x^2\\[4pt]
              x y
             \end{array}\right)=
             \left(
             \begin{array}{c}
              v\\ [4pt] 0
             \end{array}\right).
$$
Hence, ${\partial_x u}=0\Rightarrow u=\ff_1(y)$ for a certain $\ff_1:\mathbb{R}\rightarrow\mathbb{R}$. Since $\partial_x v =-u$, then $v=-\ff_1(y)x+\ff_2(y)$ for another   $\ff_2:\mathbb{R}\rightarrow \mathbb{R}$.
Using now the PDEs of the second matrix, we see that 
$xy {\partial_y u}  =xy  { \ff^\prime _1} =v=-\ff_1 x+\ff_2  $, so that  
  $\ff_2=0$ and $\ff_1=\lambda/y $, where $\lambda\in\mathbb{R}\backslash \{0\}$. It can be checked that the remaining equation is so fulfilled. Therefore $u= \lambda/y $ and $v=-\lambda x/y$.
  \end{proof}

We remark that, since  the three classes P$_2$, I$_4$ and I$_5$ appear in Table 3,   system (\ref{FirstLie}) can always be associated
to a symplectic form turning their vector fields Hamiltonian. In this respect, 
  recall that  it was recently proved, that the  system (\ref{FirstLie})  is a
 Lie--Hamilton one for any value of $c$~\cite{BCHLS13Ham}.
However,   we shall show that  identifying   it  to one of the classes of
 the GKO  classification will be useful  to study the relation of this system to other ones.


\subsection{Second-order Kummer--Schwarz equations}

Let us turn now to the second-order Kummer--Schwarz equation  written as a first-order system in the form
\begin{equation}\label{FirstLie2}
\left\{\begin{aligned}
\frac{\dd x}{\dd t}&=y,\\
\frac{\dd y}{\dd t}&=\frac 32 \frac{y^2}{x}-2c\,x^3+2b_1(t)x,
\end{aligned}\right.
\end{equation}
where $b_1(t)$ is an arbitrary $t$-dependent function and $c$ is a real constant. This equation appears in several mathematical problems and it is related to relevant differential equations appearing in physics \cite{CGL11,LA08}.

It is well known that (\ref{FirstLie2}) is a Lie system \cite{CGL11,LS12}. In fact, it describes the integral curves of the $t$-dependent vector field
$X_t=X_3+b_1(t)X_1$ where
\begin{equation}\label{KS2}
X_1=2x\frac{\partial}{\partial y},\qquad X_2=x\frac{\partial}{\partial x}+2y\frac{\partial}{\partial y},\qquad X_3=y\frac{\partial}{\partial x}+\left(\frac{3y^2}{2x}-2c\,x^3\right)\frac{\partial}{\partial y},
\end{equation}
span a Lie algebra isomorphic to $\mathfrak{sl}(2)$ with commutation rules (\ref{FirstLieB}). Thus 
$V$ can be isomorphic to one of the   four $\mathfrak{sl}(2)$-Lie algebras of vector fields     in the GKO classification.  

As in the previous subsection, we analyse  if there exists a distribution $\mathcal{D}$ stable under $V$ and locally generated by a vector field $Y$ of the first form given in (\ref{xxx}) (the same results by assuming the second one). So, impossing  $\mathcal{D}$ to be stable under $V$ yields
\begin{subequations}
\begin{align}
&\qquad\qquad\,\,\,\,\mathcal{L}_{X_1}Y=2\left(x\frac{\partial \fff_y}{\partial y} -1\right)\frac{\partial}{\partial y}=\hh_1Y\label{con2p}
,\\
&\qquad\,\, \mathcal{L}_{X_2}Y=-\frac{\partial}{\partial x}+\left(x\frac{\partial \fff_y}{\partial x}+2y\frac{\partial \fff_y}{\partial y}-2\fff_y\right)\frac{\partial}{\partial y}=\hh_2Y,\label{con1p}\\
&\mathcal{L}_{X_3}Y=-\fff_y\frac{\partial}{\partial x}+\left[X_3\fff_y+\frac{3y^2}{2x^2}+6c\,x^2-\frac{3y}x\fff_y\right]\frac{\partial}{\partial y}=\hh_3 Y\label{con3p},
\end{align}
\end{subequations}
for certain functions $\hh_1,\hh_2,\hh_3$ locally defined on $\mathbb{R}^2$.
The left-hand side of (\ref{con2p}) has no term $\partial_x$ and the right-hand one does not have it provided $\hh_1=0$. Hence, $\hh_1=0$ and $\fff_y= y/x+F$ for a $F=F(x)$. In view of (\ref{con1p}), we then obtain  $\hh_2=-1$ and 
$F -x    F^\prime=0$, that is, $F(x) =   \mm x$ for   $\mm\in\mathbb{R}$. Substituting $g_y$ in (\ref{con3p}),  
we obtain that $\hh_3=-\mm x-y/x$
and $ \mm^2=-4c$. Hence, as in the Milne--Pinney equations, we find that 
if $c>0$  it does not exist any   $Y$ spanning locally $\mathcal{D}$ satisfying (\ref{con2p})--(\ref{con3p})  and $V$ 
is primitive, meanwhile  if $c\le 0$, then   $\mathcal{D}$ is spanned by the vector field
$$
Y= \frac{\partial}{\partial x}+\left(\frac{y}{x} + {\mm}{x}\right) \frac{\partial}{\partial y},\qquad  \mm^2=-4c,
$$
and $V$ is  imprimitive.

The precise classes of the  GKO classification  corresponding to the system (\ref{FirstLie2}) are summarized in the following proposition.

\begin{proposition}\label{prop72} The system   (\ref{FirstLie2}), associated with the second-order Kummer--Schwarz    equations,    is locally diffeomorphic to   P$_2$ for $c>0$, I$_4$ for $c<0$ and I$_5$ for $c=0$.  
\end{proposition}

\begin{proof}  The case with   $c>0$ provides the primitive class P$_2$  since $Y=0$. If   $c<0$ we look for a local diffeomorphism 
$\phi:(x,y)\in U\subset \mathbb{R}^2_{x\ne y}\mapsto \bar U\subset (u,v)\in\mathbb R^2_{u\neq 0}$, such that $\phi_*$ maps the basis of  I$_4$  into  (\ref{KS2}), that is,
$$
\phi_*(\partial_x+\partial_y)=2x\partial_y,\qquad  \phi_*(x^2\partial_x+y^2\partial_y)=y\partial_x+ (\tfrac32 y^2/{x}-2c\,x^3 )\partial_y .
$$
Then
$$
\left(
\begin{array}{cc}
\frac{\partial u}{\partial x}&\frac{\partial u}{\partial y}\\[4pt]
\frac{\partial v}{\partial x}&\frac{\partial v}{\partial y}\\
\end{array}
\right)\left(\begin{array}{c}
              1\\[4pt]
              1
             \end{array}\right)=
             \left(
             \begin{array}{c}
              0\\[4pt] 2u
             \end{array}\right),\qquad
             \left(
\begin{array}{cc}
\frac{\partial u}{\partial x}&\frac{\partial u}{\partial y}\\[4pt]
\frac{\partial v}{\partial x}&\frac{\partial v}{\partial y}\\
\end{array}
\right)\left(\begin{array}{c}
              x^2\\[4pt]
              y^2
             \end{array}\right)=
             \left(
             \begin{array}{c}
              v\\ [4pt]  \tfrac 32 v^2/u -2{c}\,{u^3}
             \end{array}\right).
$$
Proceeding as in    the proof of Proposition~\ref{prop71} we find that   $u=\ff(x-y)$ and  $v=(x^2-y^2) \ff^\prime$
 for   $\ff:\mathbb{R}\rightarrow\mathbb{R}$. 
  As  now   ${\partial_x v}+{\partial_y v}=2u$ we obtain $2(x-y)\ff^\prime=2\ff$, so that $\ff=\lambda (x-y) $  with $\lambda\in\mathbb{R}\backslash \{0\}$. The remaining equation $x^2\partial_x v +y^2\partial_y v  =\tfrac 32 v^2/u -2{c}\,{u^3}$ implies that 
   $4\lambda^2=-1/ c$, which is consistent with the value $c<0$. Then
$$
u=  {\lambda } (x-y) ,\qquad v=  {\lambda }(x^2-y^2) ,\qquad 4\lambda^2=-1/ c.
$$

In the third possibility  with $c=0$  we require that $\phi_*$     maps the basis of  I$_5$  into  (\ref{KS2}) so  fulfilling
$$
\phi_*(\partial_x )=2x\partial_y,\qquad  \phi_*(x^2\partial_x+xy\partial_y)=y\partial_x+ \tfrac32 y^2/{x} \partial_y,
$$
that is,
$$
\left(
\begin{array}{cc}
\frac{\partial u}{\partial x}&\frac{\partial u}{\partial y}\\[4pt]
\frac{\partial v}{\partial x}&\frac{\partial v}{\partial y}\\
\end{array}
\right)\left(\begin{array}{c}
              1\\[4pt]
              0
             \end{array}\right)=
             \left(
             \begin{array}{c}
              0\\[4pt] 2u
             \end{array}\right),\qquad
             \left(
\begin{array}{cc}
\frac{\partial u}{\partial x}&\frac{\partial u}{\partial y}\\[4pt]
\frac{\partial v}{\partial x}&\frac{\partial v}{\partial y}\\
\end{array}
\right)\left(\begin{array}{c}
              x^2\\[4pt]
              x y
             \end{array}\right)=
             \left(
             \begin{array}{c}
              v\\ [4pt]  \tfrac 32 v^2/u
             \end{array}\right).
$$
By taking into account the proof of Proposition~\ref{prop71}, it is straightforward  to check  that the four PDEs are satisfied for
$u= \lambda y^2$ and $v=2\lambda  x y^2$ with  $\lambda\in\mathbb{R}\backslash \{0\}$.    \end{proof}


\subsection{Complex Riccati equation with $t$-dependent real coefficients}

Let us return to complex Riccati equations with $t$-dependent real coefficients in the form (\ref{Riccati2}). We already showed that   this system has a minimal Lie algebra P$_2\simeq \mathfrak{sl}(2)$. Therefore, it is locally diffeomorphic to the minimal Lie algebra appearing in the above Milne--Pinney  and Kummer--Schwarz   equations whenever the parameter $c>0$. In view of the GKO classification, there exist local diffeomorphisms relating the {\em three}  first-order systems associated with these equations. For instance, we can search for a local diffeomorphism $\phi:(x,y)\in U\subset \mathbb{R}^2_{y\ne 0}\mapsto \bar U\subset (u,v)\in\mathbb R^2_{u\neq 0}$ mapping every system (\ref{Riccati2}) into one of the form (\ref{FirstLie}), e.g.~satisfying that $\phi_*$ maps the basis (\ref{vectRiccati2}) of P$_2$, related to the planar Riccati equation, into the basis (\ref{FirstLieA}) associated with the Milne--Pinney one. By writing  in coordinates
$$
\phi_*(\partial_x)=-x\partial_y,\qquad  \phi_*[(x^2-y^2)\partial_x+2xy\partial_y]=y\partial_x+c/x^3\partial_y ,
$$
we obtain
$$
\left(
\begin{array}{cc}
\frac{\partial u}{\partial x}&\frac{\partial u}{\partial y}\\[4pt]
\frac{\partial v}{\partial x}&\frac{\partial v}{\partial y}\\
\end{array}
\right)\left(\begin{array}{c}
              1\\[4pt]
              0
             \end{array}\right)=
             \left(
             \begin{array}{c}
              0\\[4pt] -u
             \end{array}\right),\qquad
             \left(
\begin{array}{cc}
\frac{\partial u}{\partial x}&\frac{\partial u}{\partial y}\\[4pt]
\frac{\partial v}{\partial x}&\frac{\partial v}{\partial y}\\
\end{array}
\right)\left(\begin{array}{c}
              x^2-y^2\\[4pt]
              2xy
             \end{array}\right)=
             \left(
             \begin{array}{c}
              v\\ [4pt]  {c}/{u^3}
             \end{array}\right).
$$
Similar computations to those performed in the proof of Proposition~\ref{prop71} for the three PDEs 
${\partial_x u}=0$, $\partial_x v =-u$ and $(x^2-y^2) {\partial_x u}+  2xy {\partial_y u}=v$ gives
$u=\lambda/|y|^{1/2}$ and $v=-\lambda x/|y|^{1/2}$ with  $\lambda\in\mathbb{R}\backslash \{0\}$.
Substituting these results  into the remaining equation we find that $\lambda^4=c$ which is consistent with the positive value of $c$.   Consequently, this maps the system (\ref{Riccati2}) into (\ref{FirstLie}) and the solution of the first one is locally equivalent to solutions of the second one.

Summing up, we have explicitly proven that the three Lie--Hamilton  $\mathfrak{sl}(2)$ algebras of the classes  P$_2$, I$_4$ and I$_5$ given in Table 3  cover the following $\mathfrak{sl}(2)$-Lie systems:

\begin{itemize} 
 
\item P$_2$:  Milne--Pinney  and Kummer--Schwarz   equations for $c>0$ as well as complex Riccati equations with $t$-dependent coefficients.

\item I$_4$:  Milne--Pinney  and Kummer--Schwarz   equations for $c<0$.

\item I$_5$:  Milne--Pinney  and Kummer--Schwarz   equations for $c=0$ and
the harmonic oscillator with $t$-dependent frequency.

\end{itemize} 

This means that, only within each class,  they are locally diffeomorphic  and, therefore, there   exists a  
local change of variables mapping  one into another. Thus, for instance, there does not exist any diffeomorphism mapping the Milne--Pinney  and Kummer--Schwarz equations with $c\ne 0$  to the harmonic oscillator. These results also  explain from an algebraic point of view the existence of the known diffeomorphism mapping   Kummer--Schwarz equations to Milne--Pinney equations \cite{LA08} provided that the sign of $c$ is the same in both systems.


  \section{Application to biological models}

In this section  we focus on new applications of the Lie--Hamilton approach  to  Lotka--Volterra-type systems and to a viral infection model. We also consider here the analysis of Buchdahl equations which can be studied through a Lie--Hamiltonian system diffeomorphic to a precise $t$-dependent  Lotka--Volterra system.


 \subsection{Generalised Buchdahl equations}
 We call generalised Buchdahl equations~\cite{Bu64, CSL05, CN10} to the    second-order differential equations
 $$
 \frac{\dd^2 x}{\dd t^2}=a(x)\left(\frac{\dd x}{\dd t}\right)^2+b(t)\frac{\dd x}{\dd t},
 $$
where $a(x)$ and $b(t)$ are arbitrary functions of their respective arguments. In order to analyse whether these equations can be studied through a Lie system,
 we add the variable $y\equiv \dd x/\dd t $ and consider the first-order   system  of differential equations
  \begin{equation}\label{Buchdahl}
\left\{ \begin{aligned}
 \frac{\dd x}{\dd t}&=y,\\
 \frac{\dd y}{\dd t}&=a(x)y^2+b(t)y .
 \end{aligned}\right.
 \end{equation}
 Note that if $(x(t),y(t))$ is a particular solution of this system with $y(t_0)=0$ for a particular $t_0\in\mathbb{R}$, then $y(t)=0$ for every $t\in\mathbb{R}$ and $x(t)=\lambda \in\mathbb{R}$. Moreover, if $a(x)=0$ then the solution of the above system is also trivial. As a consequence, we can restrict ourselves to studying  particular solutions on $\mathbb R^2_{y\neq 0}$ with $a(x)\neq 0$.

 Next let us  prove that (\ref{Buchdahl}) is a Lie system. Explicitly,  this  is associated with the $t$-dependent vector field
 $X_t=X_1+b(t)X_2$, 
 where 
 \begin{equation}\label{VectBuch}
 X_1=y\frac{\partial }{\partial x}+a(x)y^2\frac{\partial}{\partial y},\qquad  X_{2}=y\dfrac{\partial}{\partial y}.
 \end{equation}
Since
$$
 [X_1,X_2]=-X_1,
 $$
these vector fields span a Lie algebra $V$ isomorphic to $\mathfrak{h}_2$ leaving invariant the distribution $\mathcal{D}$ spanned by $Y\equiv X_1$.  Since the rank of $\mathcal{D}^V$ is two,  $V$ is locally diffeomorphic to the imprimitive class ${\rm I}_{14A}$ with $r=1$ and $\eta_1(x)={\rm e}^x$  given in Table~\ref{table3}. This proves for the first time that generalised Buchdahl equations written as the     system (\ref{Buchdahl})  are, in fact, not only a Lie system \cite{CGL11} but a Lie--Hamilton one.

Next by determining a symplectic form obeying $\mathcal{L}_{X_i}\omega=0$, with $i=1,2$ for the vector fields   (\ref{VectBuch}) and  the  generic $\omega$ (\ref{ww}), it can be shown that this reads
\begin{equation*}
\omega=\dfrac{\exp\left(-\int a(x)\dd x\right)}{y}\,\dd x\wedge \dd y,
\end{equation*}
which turns   $\tX$ and $X_2$ into Hamiltonian vector fields with Hamiltonian functions
$$
h_1=y\exp\left(-\int^x a(x')\dd x'\right),\qquad h_2=-\int^x\exp\left(-\int^{x'} a(\bar x)d\bar x\right)dx',
$$
respectively. Note that all the these structures are properly defined on $\mathbb R^2_{y\neq 0}$ and   hold
$\{h_{1},h_{2}\}=h_{1}.
$
Consequently, the system (\ref{Buchdahl}) has a $t$-dependent Hamiltonian given by
$
h_t=h_{1}+b(t) h_{2}.
$


 \subsection{Time-dependent Lotka--Volterra systems}
 Consider the specific $t$-dependent Lotka--Volterra system~\cite{Tr96,JHL05} of the form
  \begin{equation}\label{LV}
  \left\{
 \begin{aligned}
 \frac{\dd x}{\dd t}&=ax-g(t)(x-ay)x,\\
 \frac{\dd y}{\dd t}&=ay-g(t)(bx-y)y,
 \end{aligned}\right.
\end{equation}
where $g(t)$ is a $t$-dependent function  representing the variation of the seasons and $a,b$ are certain real parameters
describing the interactions among the species. We hereafter focus on the case $a\neq 0$, as otherwise the above equation becomes, up to a $t$-reparametrization,
an autonomous differential equation that can easily be integrated. We also assume $g(t)$ to be a non-constant function and we restrict ourselves to particular solutions on $\mathbb{R}_{x,y\ne 0}=\{(x,y)|x\neq 0, y\neq 0\}$ (the remaining ones can be trivially  obtained).

Let us prove that (\ref{LV}) is a Lie system and that for some values of the real parameters $a\ne 0$ and $b$ this is   a Lie--Hamilton system as well. This system describes the integral curves of the $t$-dependent vector field $
X_t=X_1+g(t)X_2$
 where
 $$
 X_1=ax\frac{\partial}{\partial x}+ay\frac{\partial}{\partial y},\qquad X_2=-(x-ay)x\frac{\partial}{\partial x}-(bx-y)y\frac{\partial}{\partial y} ,
 $$
satisfy 
$$
[X_1,X_2]=a X_2,\qquad a\ne 0.
$$
 Hence, $X_1$ and $X_2$ are the generators of  a Lie algebra $V$ of vector fields isomorphic to $\mathfrak{h}_2$  leaving invariant the distribution $\mathcal{D}$  on $\mathbb{R}_{x,y\ne 0}$ spanned by $Y\equiv X_2$.
 According to the values of the parameters $a\ne 0$ and $b$ we find that:
 
 \begin{itemize}
 \item When  $a=b=1$,  the rank of $\mathcal{D}^V$  on the domain of $V$ is one. In view of Table \ref{table1}   the Lie algebra $V$ is thus isomorphic to I$_2$ and, by taking into account Table \ref{table3},  we conclude that $X$ is a Lie system, but not a Lie--Hamilton one.
 
 \item Otherwise,  the rank of $\mathcal{D}^V$ is two, so that this Lie algebra is locally diffeomorphic to I$_{14A}$ with $r=1$ and $\eta_1={\rm e}^{a x}$ given  in Table \ref{table3} and, consequently, $X$ is a Lie--Hamilton system.  As a straightforward consequence, when $a=1$ and $b\ne 1$ the system (\ref{LV})  is locally diffeomorphic to the  
  generalised Buchdahl equations (\ref{Buchdahl}).

 \end{itemize}

Let us now derive a symplectic structure (\ref{ww}) turning the elements of $V$ into Hamiltonian vector fields by solving the system of PDEs $\mathcal{L}_{X_1}\omega=\mathcal{L}_{X_2}\omega=0$. The first condition reads in local coordinates
$$
\mathcal{L}_{X_1}\omega=(X_1f+2af)\dd x\wedge \dd y=0.
$$
So we obtain that $f=F(x/y)/y^2$ for any function $F:\mathbb{R}\rightarrow \mathbb{R}$.
By imposing that $\mathcal{L}_{X_2}\omega=0$, we find  
$$
\mathcal{L}_{X_2}\omega=\left[(b-1)x^2+(a-1)yx\right]\frac{\partial f}{\partial x}+f\left[(b-2)x+ay\right]=0.
$$
Notice that, as expected,   $f$ vanishes when $a=b=1$. We study separately the remaining cases: i) $a\ne 1$ and  $b\ne 1$; ii) $a=1$ and $b\ne 1$; and iii) $a\ne 1$ and $b= 1$.

When  both $a,b\ne 1$  we write $f=F(x/y)/y^2$, thus obtaining that $\omega$ reads, up to a non-zero multiplicative constant, as
$$
\omega=\frac 1{y^2}\left(\frac{x}{y}\right)^{\frac a{1-a}}\left(1-a+(1-b)\frac xy\right)^{\frac 1{a-1}+\frac 1{b-1}}\dd x\wedge \dd y,\qquad a,b\neq 1.
$$
From this, we   obtain the following Hamiltonian functions for $X_1$ and $X_2$:
$$
\begin{gathered}
h_{1}=a (1-b)^{1+\frac{1}{a-1}+\frac{1}{b-1}}\left(\dfrac{x}{y}\right)^{\frac{1}{b-1}}\,_2F_1\left(
\dfrac{1}{1-b},\dfrac{1}{1-a}+\dfrac{1}{1-b};\dfrac{b-2}{b-1};\dfrac{y(1-a)}{x(b-1)}
\right),\\
\qquad h_{2}=-y \left(\dfrac{x}{y}\right)^{\frac{1}{1-a}}\left[
(1-a)+(1-b)\dfrac{x}{y}
\right]^{\frac{1}{a-1}+\frac{1}{b-1}+1},
\end{gathered}
$$
where $\,_2F_1(\alpha,\beta, \gamma,\zeta)$ stands for the well-known hypergeometric function $\,_2F_1(\alpha,\beta,\gamma,z)=\sum_{n=0}^\infty [(\alpha)_n(\beta)_n/(\gamma)_n]z^n/n!$ with $(\delta)_n=\Gamma(\delta+n)/\Gamma(\delta)$ being the rising Pochhmaler symbol. 
As expected, $\{h_1,h_2\}_\omega=-a h_2$.

When $a=1$ and $b\neq 1$, the symplectic form for $X$ becomes
$$
\omega=\frac 1{y^2}\exp\left(\frac{y-(b-2)x \ln|x/y|}{(b-1)x}\right)\dd x\wedge \dd y, \qquad b \neq 1,
$$
and the Hamiltonian functions for $X_1$ and $X_2$ read
$$
\begin{gathered}
h_{1}=-\left(\dfrac{1}{1-b}\right)^{\frac{1}{b-1}}\Gamma\left(
\dfrac{1}{1-b},\dfrac{y}{x(1-b)} 
\right),\qquad 
h_{2}=(b-1)x\left(\dfrac{x}{y}\right)^{\frac{1}{b-1}} \text{exp}\left(
\dfrac{y}{(b-1)x}
\right),
\end{gathered}
$$
with $\Gamma(u,v)$ being the incomplete Gamma function,  which satisfy  $\{h_{1}, h_{2}\}_\omega=-h_{2}$.

Finally, when $b=1$ and $a\neq 1$, we have
$$
\omega=\dfrac{1}{y^{2}}\left(
\dfrac{x}{y}
\right)^{\frac{a}{1-a}}\text{exp}\left(
\dfrac{x}{y (a-1)}
\right)\text{d} x \wedge \text{d}y ,\qquad a\neq 1.
$$
Then, the Hamiltonian functions for $X_1$ and $X_2$ are, in this order, 
$$
\begin{gathered}
h_{1}=a(1-a)^{\frac{1}{1-a}}\, \Gamma \left(\dfrac{1}{1-a},\dfrac{x}{y(1-a)}\right),\quad
h_{2}=(a-1)y\,\text{exp}\left(
\dfrac{x}{y(a-1)}
\right)\left(
\dfrac{x}{y}
\right)^{\frac{1}{1-a}} .
\end{gathered}
$$
Indeed,  $\{h_{1}, h_{2}\}_\omega=-a h_{2}$.


\subsection{Predator-prey Lie systems}

The system of differential equations \cite{LlibreValls2}
\begin{equation} \label{SLV}
\left\{
\begin{aligned}
\frac{\dd x}{\dd t}&= b(t)x + c(t)y +d(t)x^{2}+e(t)x y + f(t)y^{2},\\
\frac{\dd y}{\dd t}&= y,\\
\end{aligned}\right.
\end{equation}
where $b(t), c(t), d(t), e(t)$ and $f(t)$ are arbitrary $t$-dependent functions, is an interacting species model that belongs to the class of quadratic-linear polynomial systems with a unique singular point at the origin \cite{LlibreValls2}.

In general, this predator-prey system is not a Lie system. For instance, consider the particular system associated to the $t$-dependent vector field
$$
X_t=d(t) X_1+ e(t)X_2+X_3 ,\qquad X_1=x^2\frac{\partial}{\partial x}   ,\qquad X_2=xy\frac{\partial}{\partial x},\qquad X_3=y\frac{\partial}{\partial y},
$$
where $d(t)$ and $e(t)$ are non-constant and non-proportional functions. Notice that $V^X$ contains $X_1,X_2$ and their successive  Lie brackets, i.e. the vector fields
$$
\stackrel{n-{\rm times}}{\overbrace{[X_2,\ldots [X_2}},X_1]\ldots]= x^2 y^n \frac{\partial}{\partial x} \equiv Y_n .
$$
Hence, $[X_2,Y_n]=Y_{n+1}$ and  the family of vector fields $X_1,X_2,X_3,Y_1,Y_2,\dots$
  span an infinite-dimensional family of linearly independent vector fields over $\mathbb{R}$ so that $X$ is not a Lie system.

 Hereafter we analyse the cases of (\ref{SLV}) with $d(t)=e(t)=0$ which provides   predator--prey systems that are Lie systems. We call them {\it predator-prey Lie systems}; these are    related to the system of differential equations
  \cite{LlibreValls2}
\begin{equation} \label{SLV2}
\left\{
\begin{aligned}
\frac{\dd x}{\dd t}&=b(t)x+ c(t)y +f(t) y^{2},\\
\frac{\dd y}{\dd t}&=y.
\end{aligned}\right.
\end{equation}
Note that if a solution $(x(t),y(t))$ of the above system is such that $y(t_0)=0$ for a certain $t_0$, then $y(t)=0$ for all $t\in\mathbb{R}$ and the corresponding $x(t)$ can be then easily obtained. In view of this, we   focus on those particular solutions within $\mathbb R^2_{y\neq 0}$. The   system (\ref{SLV2})  is    associated with the $t$-dependent vector field on $\mathbb R^2_{y\neq 0}$  of the form   
$X_{t}=X_1+b(t)X_{2}+c(t) X_{3}+f(t) X_{4}$, where
\begin{equation*}
X_1=y\dfrac{\partial}{\partial y},\qquad X_2=x \dfrac{\partial}{\partial x},\qquad X_3=y\dfrac{\partial}{\partial x}, \qquad X_4=y^{2}\dfrac{\partial}{\partial x} ,
\end{equation*}
satisfy the commutation rules
$$
\begin{array}{lll}
[X_1,X_2]=0,  &\qquad [X_1,X_3]=X_3,&\qquad [X_1,X_4]=2 X_4,\\[2pt]
[X_2,X_3]=-X_3,&\qquad [X_2,X_4]=-X_4,&\qquad  [X_3,X_4]=0 .
\end{array}
$$
Note that $V\simeq V_1\ltimes V_2$ where $V_1=\langle X_1,X_2\rangle\simeq\mathbb{R}^2$ and $V_2=\langle X_3,X_4\rangle\simeq \mathbb{R}^2$. In addition, the distribution $\mathcal{D}$ spanned by $Y\equiv \partial_
x$ is invariant under the action of the above vector fields  so, $V$ is imprimitive. In view of Table \ref{table1}, we find that (\ref{SLV2}) is a Lie system corresponding to the imprimitive class I$_{15}$ with $V\simeq \mathbb{R}^2\ltimes\mathbb{R}^2$. By taking into account our classification given in Table 3, we know that this is
 not a Lie algebra of vector fields with respect to any symplectic structure.


\subsection{Predator-prey Lie--Hamilton systems}

We now   consider a subcase of (\ref{SLV2}) that provides a Lie--Hamilton system. In view of Table \ref{table3}, the Lie subalgebra $\mathbb{R}\ltimes \mathbb{R}^2$ of $V$ is a Lie algebra of Hamiltonian vector fields with respect to a symplectic structure, that is, ${\rm I}_{14}\subset {\rm I}_{15}$ as shown in Table 2. So, it is natural consider the restriction of (\ref{SLV2}) to
\begin{equation} \label{SLV3}
\left\{
\begin{aligned}
\frac{\dd x}{\dd t}&=b\, x+ c(t)y +f(t) y^{2},\\
\frac{\dd y}{\dd t}&=y,
\end{aligned}\right.
\end{equation}
where   $b\in \mathbb{R}\backslash\{1,2\}$ and $c(t)$, $f(t)$ are still   $t$-dependent functions.
The system
  (\ref{SLV3})  is    associated   to the $t$-dependent vector field
 $X_{t}=X_1+c(t) X_{2}+f(t) X_{3}$ on $\mathbb {R}^2_{y\ne 0}=\{(x,y)\in \mathbb{R}\mid y\neq 0\}$, where
\begin{equation}\label{vecsA}
X_1=b\,   x \dfrac{\partial}{\partial x}+y\dfrac{\partial}{\partial y},\qquad X_2=y\dfrac{\partial}{\partial x}, \qquad X_3=y^{2}\dfrac{\partial}{\partial x}
\end{equation}
satisfy the commutation relations
\begin{equation}\label{algG}
[\tX, X_{2}]=(1-b)X_{2}, \qquad  [\tX, X_{3}]=(2-b)X_{3}, \qquad [X_{2}, X_{3}]=0 .
\end{equation}
Therefore, the vector fields (\ref{vecsA})  generate a Lie algebra $V\simeq V_1\ltimes V_2$, where $V_1=\langle X_1 \rangle\simeq\mathbb{R}$ and $V_2=\langle X_2,X_3\rangle\simeq \mathbb{R}^2$. The domain of $V$ is $\mathbb {R}^2_{y\ne 0}$ and the rank of $\mathcal{D}^V$ is two. Moreover, the distribution $\mathcal{D}$ spanned by the vector field $Y\equiv \partial_x$ is stable under the action of the elements of $V$, which turns $V$ into an imprimitive Lie algebra. So, $V$ must be locally diffeomorphic to the imprimitive Lie algebra I$_{14}$ displayed in Table 1 for $r=2$. We already know that the class I$_{14}$ is a Lie algebra of Hamiltonian vector fields with respect to a   symplectic structure.

By imposing $\mathcal{L}_{X_i}\omega=0$ for  the vector fields (\ref{vecsA}) and  the  generic symplectic form (\ref{ww}), it can be shown that $\omega$ reads
\begin{equation*}
\omega=\dfrac{\dd x\wedge \dd y}{y^{b+1}} ,
\end{equation*}
which turns    (\ref{vecsA}) into Hamiltonian vector fields with Hamiltonian functions
$$
h_1=-\frac{x}{y^b},\qquad h_2=\frac  {   y^{1-b}}{  1-b},\qquad h_3=\frac  {   y^{2-b}}{  2-b}, \qquad b\in \mathbb{R}\backslash\{1,2\}.
$$
Note that all the above structures are properly defined on $\mathbb {R}^2_{y\ne 0}$.
The above Hamiltonian  functions span a three-dimensional Lie algebra with commutation relations
$$
\{h_{1},h_{2}\}_\omega=(b-1)h_{2}  , \qquad  \{h_{1}, h_{3}\}_\omega=(b-2)h_{3} ,\qquad  \{h_{2},h_{3}\}_\omega=0 .
$$
Consequently, $V$ is locally diffeomorphic  to the imprimitive Lie algebra I$_{14A}$ of Table 3 such that the Lie--Hamilton algebra is $\mathbb{R}\ltimes \mathbb{R}^{2}$ (also  $(\mathbb{R}\ltimes \mathbb{R}^{2})\oplus\mathbb{R}$).
The system (\ref{SLV3}) has a $t$-dependent
Hamiltonian
\begin{equation}\nonumber
h_t=b \,h_{1}+c(t) h_{2}+d(t)h_{3}=-b\dfrac{x}{y^b}+c(t) \frac  {   y^{1-b}}{  1-b} + d(t) \frac  {   y^{2-b}}{  2-b}.
\end{equation}

We point out that the cases of  (\ref{SLV3}) with either $b=1$ or $b=2$     also lead to Lie--Hamilton systems,  but  now belonging, both of them, to the   class I$_{14 B}$ of Table 3 as a central generator is  required. For instance if $b=1$, the commutation relations (\ref{algG}) reduce to
\begin{equation}\label{algG2}
[\tX, X_{2}]= 0 , \qquad  [\tX, X_{3}]=X_{3}, \qquad [X_{2}, X_{3}]=0 ,
\end{equation}
while    the symplectic form and the Hamiltonian functions are found to be
$$
\omega=\dfrac{\dd x\wedge \dd y}{y^{2}} ,\qquad h_1=-\frac{x}{y},\qquad h_2=\ln y ,\qquad h_3= y ,
$$
which together with $h_0=1$ close the   (centrally extended)  Lie--Hamilton algebra 
$\overline{\mathbb{R} \ltimes \mathbb{R}^2 }$, that is,
\begin{equation}\label{algG3}
\{h_{1},h_{2}\}_\omega=-h_0 , \qquad  \{h_{1}, h_{3}\}_\omega=-h_{3} ,\qquad  \{h_{2},h_{3}\}_\omega=0 ,\qquad  \{h_{0},\cdot\}_\omega=0 .
\end{equation}
A similar result can be found for $b=2$.


 \subsection{A primitive model of viral infection}

Finally, let us consider a simple viral infection model given by \cite{EK05}
\begin{equation}\left\{
\begin{aligned}
\dfrac{{\rm d}x}{{\rm d}t}&=(\alpha(t)-g(y))x,\\
\dfrac{{\rm d}y}{{\rm d}t}&=\beta(t) x y-\gamma(t) y,
\end{aligned}\label{VS1}\right.
\nonumber
\end{equation}
where $g(y)$ is a real positive function taking into account the power of the infection. Note that if a particular solution satisfies $x(t_0)=0$ or $y(t_0)=0$ for a $t_0\in\mathbb{R}$, then
$x(t)=0$ or $y(t)=0$, respectively, for all $t\in\mathbb{R}$. As these cases are trivial, we restrict
ourselves to studying particular solutions within $\mathbb R^2_{x,y\ne 0}=\{(x,y)\in \mathbb{R}^2\mid x\neq 0,  y\neq 0\}$.

The simplest possibility consists in setting  $g(y)=\delta$, where $\delta$ is a constant.
 Then, (\ref{VS1})    describes the integral curves of the $t$-dependent vector field
$X_t=(\alpha(t)-\delta)X_{1}+\gamma(t) X_{2} +\beta(t) X_{3}$, 
on $\mathbb R^2_{x,y\ne 0}$, where  the vector fields
\begin{equation*}
X_1=x\dfrac{\partial}{\partial x}, \qquad X_2=- y \dfrac{\partial}{\partial y}, \qquad X_3=x y \dfrac{\partial}{\partial y} ,
\end{equation*}
satisfy the   relations (\ref{algG2}).
 So, $X$ is a Lie system related to a Vessiot--Guldberg Lie algebra $V\simeq \mathbb{R}\ltimes\mathbb{R}^2$ where
  $ \langle X_1 \rangle\simeq\mathbb{R}$ and $ \langle X_2,X_3\rangle\simeq \mathbb{R}^2$. The distribution $\mathcal{D}^V$ has rank two on $\mathbb R^2_{x,y\ne 0}$. Moreover, $V$ is imprimitive, as the distribution $\mathcal{D}$ spanned by $Y\equiv \partial_y$ is invariant under the action of vector fields of $V$. Thus $V$ is locally diffeomorphic to the
     imprimitive Lie algebra I$_{14 B}$  for $r=2$ and, in view of Table \ref{table3}, the system $X$ is
       a Lie--Hamilton one.

 Next we   obtain that $V$ is a Lie algebra of Hamiltonian vector fields with respect to the symplectic form
 \begin{equation*}
\omega=\dfrac{\dd x \wedge \dd y}{x y}.
\end{equation*}
Then, the vector fields $X_1$, $X_2$ and $X_3$ have Hamiltonian functions
$
h_1=\ln y, h_2=\ln x, h_3=-x,
$
which along $h_0=1$ close the   relations (\ref{algG3})  defining the  Lie--Hamilton algebra 
$\overline{(\mathbb{R} \ltimes \mathbb{R}^2 ) }$. 
If we assume $V^X=V$, the $t$-dependent Hamiltonian $
h_t=(\alpha(t) -\delta)h_1+\gamma(t) h_2+\beta(t)h_3
$
gives rise to a Lie--Hamiltonian structure $(\mathbb R^2_{x,y\ne 0},\omega, h)$ for $X$.


\section{Discussion and open problems}
We have determined which Lie algebras of the GKO  classification correspond to Hamiltonian vector fields with respect to a Poisson structure around a generic point. We found that only eleven    of the initial 28 classes of finite-dimensional Lie algebras on the plane obtained by GKO are of this type. In turn, these classes give rise to {\em  twelve} families of Lie algebras of Hamiltonian vector fields. This led to classifing Lie--Hamilton systems on the plane.

To illustrate our results, we have studied some new Lie and  Lie--Hamilton systems of interest that belong to the classes P$_2$,  I$_2$,  I$_4$,  I$_5$,  I$_{14A}$, I$_{14B}$ and I$_{15}$. In particular, our classification has been used to show that Kummer--Schwarz, Milne--Pinney equations (both with $c>0$) and complex Riccati equations with $t$-dependent coefficients are related to
the same Lie--Hamilton algebra P$_2$, a fact which was used to explain the existence of a local diffeomorphism that maps each one of these systems into another. We also showed that
the  $t$-dependent harmonic oscillator, arising from Milne--Pinney equations when $c=0$,    corresponds to class   I$_5$ and   this can  only    be related through   diffeomorphisms  to the 
 Kummer--Schwarz equations with $c=0$, but not with complex Riccati equations.

The new Lie and  Lie--Hamilton systems analised in this work contribute to enlarge the applications of Lie systems. We still aim to identify other relevant models through Lie--Hamilton systems and we plan to derive superposition rules for all Lie--Hamilton systems on the plane by applying the algebraic method devised in \cite{BCHLS13Ham}, which makes use of   Casimir functions and Poisson coalgebra structures. Additionally, there are plenty of Lie systems on the plane with polynomial quadratic coefficients. We plan to study the existence and the maximal number of their cyclic limits so as to investigate the so-called second Hilbert's number ${\rm H}(2)$ for these systems \cite{Li03}. This could be a first step to analize the XVI Hilbert's problem through our Lie techniques. Work on these lines is currently in progress.

\section*{Acknowledgements}
J. de Lucas and C. Sard\'on acknowledge partial financial support by
research projects  MTM--2009--11154 (MEC) and the Polish National Science Centre grant
under the contract number DEC-2012/04/M/ST1/00523.
A.~Ballesteros, A.~Blasco and F.J.~Herranz acknowledge partial financial support from
  the Spanish MINECO  under grant    MTM2010-18556 (with EU-FEDER support). C. Sard\'on acknowledges a fellowship provided by the
University of Salamanca.


\end{document}